\newcommand{\longversion}[1]{#1}
\newcommand{\shortversion}[1]{}

\documentclass[sigconf]{aamas} 

\usepackage{amsmath}
\usepackage{tikz}
\usepackage{makecell}
\usepackage{multirow}
\usepackage{algorithmic}
\usepackage{algorithm}

 \renewenvironment{proof}{\noindent{\sc Proof:}}{ \hfill $\square$\\ }
\allowdisplaybreaks

\usepackage{cancel}

\usepackage{fullpage}
\usepackage{tikz}
\usetikzlibrary{arrows,positioning}

\newdimen\prevdp
\def\leftlabel#1{\noalign{\prevdp=\prevdepth
   \kern-\prevdp\nointerlineskip\vbox to0pt{\vss\hbox{\ensuremath{#1}}}\kern\prevdp}}

\usepackage{url}

\usepackage{enumitem}
\usepackage{tikz}

\usepackage{xspace}

\newcommand{\NPC}{\ensuremath{\mathsf{NP}}\text{-complete}\xspace}
\newcommand{\NPH}{\ensuremath{\mathsf{NP}}\text{-hard}\xspace}

\newcommand{\suc}{\ensuremath{\succ}\xspace}

\newcommand{\WTH}{\ensuremath{\mathsf{W[2]}}-hard\xspace}
\newcommand{\FPT}{\ensuremath{\mathsf{FPT}}\xspace}

\newcommand{\SBON}{{\sc Shift Bribery over Social Network}\xspace}
\newcommand{\DSP}{{\sc Dominating set}\xspace}
\newcommand{\ktDSP}{{($k$,$t$)-\sc Dominating set}\xspace}

\let\oldlambda\lambda
\renewcommand{\lambda}{\ensuremath{\oldlambda}\xspace}
\let\oldalpha\alpha
\renewcommand{\alpha}{\ensuremath{\oldalpha}\xspace}
\let\oldDelta\Delta
\renewcommand{\Delta}{\ensuremath{\oldDelta}\xspace}

\newcommand{\tw}{\text{tw}\xspace}

\newcommand{\yes}{{\tt yes}\xspace}

\renewcommand{\AA}{\ensuremath{\mathcal A}\xspace}

\newcommand{\CC}{\ensuremath{\mathcal C}\xspace}

\newcommand{\FF}{\ensuremath{\mathcal F}\xspace}
\newcommand{\GG}{\ensuremath{\mathcal G}\xspace}

\newcommand{\LL}{\ensuremath{\mathcal L}\xspace}

\newcommand{\OO}{\ensuremath{\mathcal O}\xspace}
\newcommand{\PP}{\ensuremath{\mathcal P}\xspace}
\newcommand{\QQ}{\ensuremath{\mathcal Q}\xspace}

\renewcommand{\SS}{\ensuremath{\mathcal S}\xspace}

\newcommand{\UU}{\ensuremath{\mathcal U}\xspace}
\newcommand{\VV}{\ensuremath{\mathcal V}\xspace}

\newcommand{\XX}{\ensuremath{\mathcal X}\xspace}

\newcommand{\ppp}{\ensuremath{\mathfrak p}\xspace}

\newcommand{\sss}{\ensuremath{\mathfrak s}\xspace}

\newcommand{\RB}{\ensuremath{\mathbb R^+}\xspace}

\usepackage{nicefrac}

\newtheorem{observation}{\sc Observation}[section]
\newtheorem{theorem}{\sc Theorem}[section]
\newtheorem{lemma}{\sc Lemma}[section]
\newtheorem{corollary}{\sc Corollary}[section]
\newtheorem{definition}{\sc Definition}[section]
\newtheorem{claim}{\sc Claim}[section]

\newtheorem{probdefinition}{\sc Problem Definition}[section]

\newcommand{\eps}{\ensuremath{\varepsilon}\xspace}
\renewcommand{\epsilon}{\eps}

\newcommand{\ignore}[1]{}

\newcommand{\pr}{\ensuremath{\prime}}

\renewcommand{\leq}{\leqslant}
\renewcommand{\geq}{\geqslant}
\renewcommand{\ge}{\geqslant}
\renewcommand{\le}{\leqslant}

\usepackage{cleveref}

\crefname{theorem}{Theorem}{\bf Theorems}
\crefname{observation}{Observation}{\bf Observations}
\crefname{lemma}{Lemma}{\bf Lemmata}
\crefname{corollary}{Corollary}{\bf Corollaries}
\crefname{proposition}{Proposition}{\bf Propositions}
\crefname{definition}{Definition}{\bf Definitions}
\crefname{claim}{Claim}{\bf Claims}
\crefname{reductionrule}{Reduction rule}{\bf Reduction rules}
\crefname{probdefinition}{Problem Definition}{Problem Definitions}

\usepackage{todonotes}
\usepackage{amsthm}

\usepackage{balance} 
\makeatletter
\newenvironment{proofsketch}{%
  \par
  \pushQED{\qed}%
  \normalfont \topsep6\p@\@plus6\p@\relax
  \trivlist
  \item[\@proofindent\hskip\labelsep
        {\scshape Proof Sketch\@addpunct{.}}]\ignorespaces
}{%
  \popQED\endtrivlist\@endpefalse
}
\makeatother

\setcopyright{ifaamas}
\copyrightyear{2026}
\acmYear{2026}
\acmDOI{}
\acmPrice{}
\acmISBN{}

\acmSubmissionID{502}

\title[AAMAS-2026 Formatting Instructions]{Shift Bribery over Social Networks}

\author{Ashlesha Hota}
\affiliation{
  \institution{IIT Kharagpur}
  \city{Kharagpur}
  \country{India}}
\email{ashleshahota@gmail.com}

\author{Susobhan Bandopadhyay}
\affiliation{
  \institution{TIFR Mumbai}
  \city{Mumbai}
  \country{India}}
\email{susobhanbandopadhyay@gmail.com}

\author{Palash Dey}
\affiliation{
  \institution{IIT Kharagpur}
  \city{Kharagpur}
  \country{India}}
\email{palash.dey@cse.iitkgp.ac.in}

\begin{abstract}
In shift bribery, a briber seeks to promote his preferred candidate by paying voters to raise their ranking. Classical models of shift bribery assume voters act independently, overlooking the role of social influence. However, in reality, individuals are social beings and are often represented as part of a social network, where bribed voters may influence their neighbors, thereby amplifying the effect of persuasion. We study \SBON, where voters are modeled as nodes in a directed weighted graph, and arcs represent social influence between them. In this setting, bribery is not confined to directly targeted voters—its effects can propagate through the network, influencing neighbors and amplifying persuasion. Given a budget and individual cost functions for shifting each voter's preference toward a designated candidate, the goal is to determine whether a shift strategy exists—within budget—that ensures the preferred candidate wins after both direct and network-propagated influence takes effect.

We show that the problem is \NPC even with two candidates and unit costs, and \WTH when parameterized by budget or maximum degree. On the positive side, we design polynomial-time algorithms for complete graphs under plurality and majority rules and path graphs for uniform edge weights, linear-time algorithms for transitive tournaments for two candidates, linear cost functions and uniform arc weights, and pseudo-polynomial algorithms for cluster graphs. We further prove the existence of fixed-parameter tractable (\FPT) algorithms with treewidth as parameter for two candidates, linear cost functions and uniform arc weights and pseudo \FPT with cluster vertex deletion number for two candidates and uniform arc weights. 
Together, these results give a detailed complexity landscape for shift bribery in social networks.

\end{abstract}


\newcommand{\BibTeX}{\rm B\kern-.05em{\sc i\kern-.025em b}\kern-.08em\TeX}

\begin{document}


\pagestyle{fancy}
\fancyhead{}


\maketitle 

\section{Introduction}

Elections are a fundamental mechanism for aggregating the preferences of a group of individuals into a collective decision. They are used in contexts ranging from political elections and committee decisions to multi-agent systems and recommender platforms. 
The design and analysis of elections has long been a central theme in social choice theory, 
and more recently in computational social choice, where one studies the algorithmic aspects 
of voting and the vulnerabilities of election systems to strategic behavior.  

One of the most studied forms of strategic behavior is bribery. In bribery problems, an external agent sometimes called a briber attempts to change the election outcome by paying selected voters to alter their preferences. The cost of a bribery depends on the voting rule and the way in which preferences are changed. A particularly natural and well-studied variant is Shift Bribery. Here the briber seeks to promote a distinguished candidate $c$ by moving them higher in some voters' preference orders. The cost of bribing a voter is proportional to the distance that the preferred candidate is shifted. This model is appealing as it abstracts realistic campaign strategies: a candidate cannot completely rewrite voters' rankings, but may be able to persuade them to view her more favorably and rank her slightly higher. Shift bribery has been shown to be computationally challenging in many settings, but also admits tractable algorithms under certain parameterizations.  

However, the classical model of shift bribery treats voters as independent agents. Each bribed voter changes their ranking in isolation, and the effect of a bribery stops there. This assumption is unrealistic in many real-world settings. In practice, individuals are socially embedded: they interact with colleagues, friends, and neighbors, and these interactions shape their opinions and choices. A voter who changes her preference may influence her peers to reconsider their own preferences, even if those peers were not directly targeted by the briber. For example, in political campaigns, convincing a few well-connected individuals may have cascading effects in their communities. In marketing, targeting influential customers may lead to broader adoption of a product through word-of-mouth.  

To capture such phenomena, it is natural to study bribery \emph{over social networks}. 
In this setting, the set of voters is represented by the vertices of a graph, with edges modeling social ties and possibly weighted by their strength. When a voter is bribed to shift the preferred candidate forward, this shift can partially propagate through her neighbors in the network. Thus, bribing a small number of carefully chosen individuals can have an impact far beyond those voters themselves. This network-aware perspective combines ideas from classical bribery models with those from influence maximization and diffusion processes in networks.  

Studying bribery over networks is important both from a theoretical and a practical 
perspective. It also offers a richer and more realistic abstraction of how persuasion and campaigning work in social contexts. Our work initiates a systematic exploration of this problem, which we call \SBON, and provides the foundations for understanding how social influence interacts with strategic interference in elections.

\subsection{Contributions}\label{sec:contributions}

We provide a comprehensive algorithmic and complexity-theoretic study of the \SBON problem, positioning it within the broader landscape of combinatorial optimization problems on graphs. Our main contributions can be summarized as follows.
\begin{enumerate}
    \item \textbf{Complexity Analysis for General Graphs.} We show that \SBON is \NPC for two candidates and identity cost functions on general undirected graphs. Moreover, we prove \WTH hardness when parameterized by either the bribery budget or the maximum degree of the graph, even for connected graphs with two candidates, unit costs, and uniform edge weights (\Cref{w2hbudget}, \Cref{w2hdegree}). We further extend the hardness to bipartite and directed acyclic graphs via a parameterized reduction from \textsc{Set Cover}.

    \item \textbf{Polynomial-Time Algorithms for Special Graph Classes.} For complete graphs, we show that \SBON admits polynomial-time algorithms under majority and plurality rules with linear cost functions and uniform edge weights. We show a pseudo poly.  time algorithm for uni-directed paths. On transitive tournaments, we design a linear-time algorithm for the two-candidate case with uniform costs and unit edge weights.
    For cluster graphs we design pseudo-polynomial time algorithm.

    
    \item \textbf{Fixed-Parameter Tractability.} We prove that \SBON is FPT when parameterized by the number of additional supporters required to reach a majority, via a reduction to the classical \ktDSP problem. We also show pseudo-\FPT algorithms when parameterized by the cluster vertex deletion number, demonstrating tractability on graph classes close to cluster graphs. We further develop a dynamic programming algorithm over a nice tree decomposition of the underlying graph with bounded treewidth, yielding an \FPT algorithm for \SBON for two candidates, linear cost functions and uniform edge weights.

\end{enumerate}

Overall, our work provides a comprehensive landscape of the computational complexity of \SBON across graph topologies and parameters, identifying both hardness boundaries and tractable cases that can inform practical algorithms in networked election settings. \Cref{tab:parameterized-hardness} and \Cref{tab:complexity-analysis} summarize our results.

\begin{table*}[t]
\centering
\caption{Parameterized Complexity of \SBON with two candidates, uniform edge weights and unit cost.}
\begin{tabular}{|>{\centering\arraybackslash}p{0.15\linewidth}|>{\centering\arraybackslash}p{0.1\linewidth}|>{\centering\arraybackslash}p{0.15\linewidth}|>{\centering\arraybackslash}p{0.15\linewidth}|>{\centering\arraybackslash}p{0.1\linewidth}|>{\centering\arraybackslash}p{0.1\linewidth}|>{\centering\arraybackslash}p{0.1\linewidth}|}
\hline
\multicolumn{3}{|c|}{\textbf{Election Parameters}} & 
\multicolumn{4}{|c|}{\textbf{Graph Parameters}} \\
\hline
\#candidates & \#voters & Budget & 
Degree & Treewidth & CVD & FVS \\
\hline
\WTH (\Cref{candidates-hard})& \FPT (\Cref{affected-voters-fpt}) & \WTH (\Cref{w2hbudget}) & \WTH (\Cref{w2hdegree}) & \FPT (\Cref{twproof})  & \FPT (\Cref{cvd-fpt}) & \FPT (\Cref{fvs})  \\
\hline
\end{tabular}
\label{tab:parameterized-hardness}
\end{table*}

\begin{table}[t]
\centering
\caption{Complexity results for \SBON. All instances assume edge weights $1$ and identity cost functions, except for entries marked with $^\star$ (arbitrary edge weights) and $^\dagger$ (linear cost functions). Here $r$ denotes number of cliques  in the input graph and $b$ denotes the budget.}
\begin{tabular}{|c|c|c|}
\hline
\#\textbf{Candidates} & \textbf{Graph class} & \textbf{Complexity Status} \\
\hline
$m=2$ & complete graph$^\star$ & \NPC \\
\hline
$m=2$ & connected graph & \NPC \\
\hline
$m=2$ & bipartite & \NPC \\
\hline
$m=2$ & transitive tournament &  $\OO(n)$ algorithm  \\
\hline
arbitrary $m$ & clique$^\dagger$ & $\OO(n \log n)$ algorithm  \\
\hline
$m=2$ & cluster graph & $O(r\cdot b)$ algorithm\\
\hline
arbitrary $m$ & directed path$^\dagger$ & $\OO(n^2 \cdot b)$ algorithm \\
\hline
\end{tabular}
\label{tab:complexity-analysis}
\end{table}

\subsection{Related Work}
The study of manipulative actions in elections, including bribery and control, is a central theme in computational social choice \cite{faliszewski2016control}. 
Bribery was formally introduced by Faliszewski et al.~\cite{faliszewski2009hard}, modeling situations where an external agent influences voters’ preferences within a fixed budget. 
Among the variants of bribery, \emph{Shift Bribery}---introduced by Elkind et al.~\cite{elkind2009swap}---has received considerable attention as it captures campaign management scenarios where a preferred candidate is promoted in the rankings by shifting them forward.

Early complexity results demonstrated that Constructive Shift Bribery is \NPH for rules such as Borda, Copeland, and Maximin, but polynomial-time solvable for $k$-Approval and Bucklin rules \cite{bredereck2016prices, elkind2009swap}. 
Approximation algorithms have also been proposed: Elkind and Faliszewski~\cite{DBLP:conf/wine/ElkindF10} provided a $2$-approximation for scoring rules, later extended to polynomial-time approximation schemes by Faliszewski et al.~\cite{elkind2009swap, faliszewski2021approximation}.


Further extensions include analyzing bribery in \emph{iterative voting systems}, where candidates are eliminated in rounds. 
Maushagen et al.~\cite{maushagen2018iterative} proved that both constructive and destructive shift bribery are \NPC for prominent iterative rules such as Hare, Coombs, Baldwin, and Nanson. 
Their results were recently extended and refined in a journal version \cite{maushagen2022iterative}, establishing hardness across a broader range of elimination-based systems.

Shift Bribery has also been studied in the context of \emph{multiwinner elections}. 
Bredereck et al.~\cite{bredereck2021committee} investigated rules such as SNTV, Bloc, $k$-Borda, and Chamberlin--Courant, showing that shift bribery is typically harder in multiwinner settings than in single-winner cases. 
This result emphasizes the increased complexity when the goal is to influence committee selections rather than individual winners.

Parameterized complexity analysis has provided further insights. 
Bredereck et al.~\cite{bredereck2016prices} demonstrated that the complexity of Shift Bribery depends heavily on the chosen parameterization and the class of price functions. 
For example, the problem is \WTH when parameterized by the number of affected voters for Borda, Maximin, and Copeland, but becomes fixed-parameter tractable for parameters such as the number of unit shifts under certain rules.

\section{Preliminaries and Problem Definitions}\label{sec:prelim}

An election is a pair $(\CC, \VV)$, where $\CC = \{c_1,\ldots,c_m\}$ is the set of candidates and $\VV = \{v_1 , \ldots, v_n\}$ is the set of voters. Unless stated otherwise, we use $n$ and $m$ to denote the number of voters and candidates, respectively. Each voter $v_i$ has a 
preference order (vote) $\succ_i$, which is a linear order over $\CC$.
We denote the set of all complete orders over $\CC$ by $\LL(\CC)$. A list of $n$ preference orders $\{\succ_1, \succ_2, \ldots,\succ_n\} \in \LL(\CC)^n$ is called an \emph{$n$-voter preference profile}. We denote the $i^\text{th}$ preference order of a preference profile $\PP$ by $\succ_i^\PP$. 

A map $r : \LL(\CC)^n \rightarrow \CC$ is called a \emph{resolute voting rule} (as we assume the unique-winner model). In case of ties, the winner is determined using a \emph{lexicographic tie-breaking order} $\succ_t$, which is a fixed total order over $\CC$. For a set of candidates $X$, let $\overrightarrow{X}$ denote an ordering over $X$, and $\overleftarrow{X}$ denote its reverse.

For any positive integer $\ell$, let $[\ell]$ denote the set $\{1, 2, \ldots, \ell\}$.  
A voting rule $r$ is called \emph{anonymous} if for every preference profile $(\succ_i )_{i \in [n]} \in \LL(\CC)^n$ and every permutation $\sigma$ of $[n]$, we have $r((\succ_i )_{i \in [n]} ) = r((\succ_{\sigma(i)} )_{i \in [n]} )$.  
A voting rule is called \emph{efficient} if the winner can be computed in time polynomial in the input size.

A \emph{scoring rule} is induced by an $m$-dimensional vector $(\alpha_1, \ldots, \alpha_m) \in \mathbb{Z}^m$ with $\alpha_1 \geq \alpha_2 \geq \ldots \geq \alpha_m$ and $\alpha_1 > \alpha_m$.  
A candidate receives a score of $\alpha_i$ from a voter if she is placed at the $i^\text{th}$ position in that voter’s preference order. The total score of a candidate is the sum of her scores across all voters.  


We use $s(c)$ to denote the total score of a candidate $c \in \CC$.  
The voting rule under consideration will be clear from the context.  
We assume that the briber has full knowledge of the voters’ preferences.

We now define our problems formally. Let $r$ be a voting rule. We denote a graph by $\GG = (\VV, E)$, where $\VV$ is a non-empty set of vertices on the voter set $V$ and $E$ is the set of edges. In a directed graph, each edge in $E$ carries an orientation and is called an \emph{arc}. For an undirected graph, the open neighborhood of a vertex $v$ is given by $N_\GG(v) = \{u \mid \{u,v\} \in E\}$, and the closed neighborhood is $N_\GG[v] = N_\GG(v) \cup \{v\}$. For a directed graph and a vertex $v \in \VV$, the out-neighborhood is defined as $N_\GG^+(v) = \{u \mid (v,u) \in E\}$ and the in-neighborhood as $N_\GG^-(v) = \{u \mid (u,v) \in E\}$. When the graph is clear from context, we omit the subscript $\GG$. Given a set of vertices $\SS \subseteq \VV$, the subgraph of $\GG$ induced by $\SS$ is denoted by $\GG[\SS]$.

Let $w: E \to \mathbb{N}_{\ge 0}$ be a function assigning non-negative weights to the edges of the graph. A graph is called \emph{uniform} if each edge has weight equal to 1, and \emph{general} if weights take arbitrary non-negative values. For brevity, we omit preliminaries on parameterized complexity here; they can be found in the Supplementary Material.

A \emph{shift vector} $\sss = (s_1, \ldots, s_n) \in \mathbb{N}_0^n$ specifies, for each voter $v_i \in \VV$, the number of positions the preferred candidate $c$ is shifted forward (to the left) in $v_i$’s preference order. The value $s_i$ thus denotes the \emph{initial shift} directly applied to voter $v_i$ by the briber.

Each voter $v_i$ is associated with a \emph{cost function} $\pi_i : [m-1] \rightarrow \mathbb{N}$, where $\pi_i(s_i)$ represents the cost of shifting the preferred candidate $c$ by $s_i$ positions in $v_i$’s preference order. We assume $\pi_i(0) = 0$ for all $i \in [n]$.  
When all cost functions are identical and linear, i.e., $\pi_i(s_i) = s_i$ for all $i \in [n]$, we say that the cost functions are \emph{identity cost functions}. Let $pos_i(c)$ denote the initial position (rank) of the preferred candidate $c$ in the preference order $\succ_i^{\PP}$ of voter $v_i$, where $pos_i(c)=1$ indicates that $c$ is ranked first.

\begin{definition}[Effect of Applying Shift Vector on a Network of Voters]
	Suppose we have a set $\CC$ of $m$ candidates, a set $\VV$ of $n$ voters, an $n$-voter profile $\PP = (\succ_i^{\PP})_{i \in [n]} \in \LL(\CC)^n$ over $\CC$, a directed weighted network $\GG=(\VV,\AA,w:\AA\longrightarrow\RB)$, and a shift vector $\sss = (s_1, \dots, s_n) \in \mathbb{N}_0^n$. We define a tuple $(s_1^\pr, \dots, s_n^\pr) \in \mathbb{N}_0^n$ as follows.
    
	\[ s_i^\prime = s_i +  \sum_{j\in[n]\setminus\{i\}: (j,i)\in\AA} \left \lfloor s_j\cdot w(j,i)\right \rfloor \]

	Let $\suc_i^\QQ$ be the preference obtained from $\suc_i^\PP$ by shifting $c$ to left by $\min\{s_i^\pr,$ pos$(c)\}$ positions. We call the profile $\QQ = (\succ_i^{\QQ})_{i \in [n]}$ to be the profile resulting from applying \sss on \GG.
\end{definition}
\begin{probdefinition}[\SBON]\label{def:sbon}
	Given a set $\CC$ of $m$ candidates, a set $\VV$ of $n$ voters, an $n$-voter profile $\PP = (\succ_i^{\PP})_{i \in [n]} \in \LL(\CC)^n$, a directed weighted network $\GG=(\VV,\AA,w:\AA\to\mathbb{R}_{\ge 0})$, a preferred candidate $c \in \CC$, a family $\Pi = (\pi_i : [m - 1] \to \mathbb{N})_{i \in [n]}$ of cost functions (where $\pi_i(0) = 0$ for all $i \in [n]$), and a budget $b \in \mathbb{R}_{\ge 0}$,  
	the task is to determine whether there exists a shift vector $\sss = (s_1, \dots, s_n) \in \mathbb{N}_0^n$ such that:
	\begin{enumerate}
	    \item $\sum_{v_i \in \VV} \pi_i(s_i) \le b$, and
	    \item $c$ is a winner in the profile obtained by applying $\sss$ on $\GG$.
	\end{enumerate}
\end{probdefinition}
An instance of \SBON is denoted by $(\CC, \PP, \GG, c, \Pi, b)$.
\section{Results: Classical Complexity}\label{res-class-comp}

We begin our study of \SBON by analyzing its complexity on general graphs and showing strong hardness results even under severe restrictions. While these results highlight the inherent difficulty of the problem, many real-world social networks exhibit additional structure, motivating the study of special graph classes. For instance, complete graphs capture highly interconnected communities where every individual interacts with everyone else, cluster graphs represent societies partitioned into cohesive groups, and tournaments model hierarchical or dominance-based relations. Bipartite networks arise naturally in two-layered systems such as influencer–follower platforms, while graphs with bounded treewidth reflect sparse or tree-like communities.

Investigating these classes is both theoretically and practically relevant. From a theoretical perspective, they allow us to pinpoint the boundary between tractable and intractable instances. From a practical standpoint, they reflect natural patterns of social interaction that shape how influence and persuasion spread in elections.

\subsection{General graph network}
We show that \SBON is \NPC even with two candidates and identity cost functions under the majority voting rule, in contrast to the classical shift bribery problem, which is solvable in polynomial time for both majority and plurality rules. We reduce from \DSP which is known to be \NPC.

\begin{definition}[\DSP]\label{def:dsp}
Given an undirected graph $\GG = (V, E)$, a subset of vertices D is called a dominating set if for every vertex $v \in V \backslash D$, there exists a vertex $u \in D$ such that $(u, v) \in E$.
\end{definition}

\begin{theorem}[$\star$]\label{sbon-npc}
	\SBON is \NPC when all members of the family of cost functions $\Pi$ are identity functions and there are 2 candidates.
\end{theorem}
\begin{proofsketch}
We prove this by a polynomial-time reduction from \DSP, which is known to be \NPC. Given an instance $(\GG, k)$ of \DSP, we construct, in polynomial time, an equivalent instance $(\CC, \PP, \GG_1, c, \Pi, b)$ of \SBON.

Let $\CC = \{c_1, c_2\}$ be the set of candidates, where $c_2$ is the preferred candidate $c$ of the briber. Each vertex $v_i \in \VV$ of $\GG$ corresponds to a voter $v_i$ in the election. We define a directed network $\GG_1 = (\VV_1, E_1, w)$ as follows. We define $\VV_1= \VV \cup \VV_I$, where $\VV_I$ induces an independent set of size $n-1$ in $\GG_1$, and $E_1=E$. Finally, we put weight one on each edge in $E$.
Every voter strictly prefers $c_1$ to $c_2$, i.e., for all $v_i \in \VV_1$, 
$c_1 \succ_i c_2$, at the beginning. Each voter $v_i$ has an identity cost function $\pi_i(s_i) = s_i$, meaning that shifting $c_2$ left by one position (to the top) in $v_i$’s preference order $\succ^{\PP}_{i}$ costs exactly one unit. We set the total bribery budget $b = k$. 

Hence, the constructed instance $(\CC, \PP, \GG_1, c_2, \Pi, b)$ of \SBON can clearly be computed in polynomial time from $(\GG, k)$.
We now argue that the instance $(\GG, k)$ of \DSP is a \yes-instance if and only if the constructed instance $(\CC, \PP, \GG_1, c_2, \Pi, b)$ of \SBON is a \yes-instance.
\end{proofsketch}
\longversion{
\begin{proof}
We prove this by a polynomial-time reduction from \DSP, which is known to be \NPC. Given an instance $(\GG, k)$ of \DSP, we construct, in polynomial time, an equivalent instance $(\CC, \PP, \GG_1, c, \Pi, b)$ of \SBON.

Let $\CC = \{c_1, c_2\}$ be the set of candidates, where $c_2$ is the preferred candidate $c$ of the briber. Each vertex $v_i \in \VV$ of $\GG$ corresponds to a voter $v_i$ in the election. We define a directed network $\GG_1 = (\VV_1, E_1, w)$ as follows. We define $\VV_1= \VV \cup \VV_I$, where $\VV_I$ induces an independent set of size $n-1$ in $\GG_1$, and $E_1=E$. Finally, we put weight one on each edge in $E$.


Every voter strictly prefers $c_1$ to $c_2$, i.e., for all $v_i \in \VV_1$, 
$c_1 \succ_i c_2$, at the beginning. Each voter $v_i$ has an identity cost function $\pi_i(s_i) = s_i$, meaning that shifting $c_2$ left by one position (to the top) in $v_i$’s preference order $\succ^{\PP}_{i}$ costs exactly one unit. We set the total bribery budget $b = k$. 

Hence, the constructed instance $(\CC, \PP, \GG_1, c_2, \Pi, b)$ of \SBON can clearly be computed in polynomial time from $(\GG, k)$.
We now argue that the instance $(\GG, k)$ of \DSP is a \yes-instance if and only if the constructed instance $(\CC, \PP, \GG_1, c_2, \Pi, b)$ of \SBON is a \yes-instance.

$(\Rightarrow)$ Suppose $(\GG, k)$ is a \yes-instance of \DSP, with $D \subseteq \VV$ as solution of size at most $k$. We construct a bribery strategy by bribing exactly the voters corresponding to vertices in $D$, shifting $c_2$ to the top of their preference vectors (i.e., $s_i = 1$ for all $v_i \in D$ and $s_i = 0$ otherwise). Since $\pi_i(s_i) = s_i$ and $|D| \le b$, the total bribery cost is within the budget.

Since $D$ is a dominating set in $\GG$, every vertex in $\VV \setminus D$ has at least one bribed neighbor. Moreover, each edge in $\GG_1$ (same as the edges in $\GG$) has weight one, hence, for each unbribed voter $v_j$, the accumulated influence from its bribed neighbors ensures that its effective shift $s_j' \ge 1$, making $c_2$ move to the top of its preference vector. Consequently, every voter in the original graph $\GG$ now ranks $c_2$ above $c_1$. As there are $n$ such voters in $\GG$ and $2n - 1$ voters in total (including the isolated ones), $c_2$ receives at least $n$ votes and becomes the winner. Hence, the constructed instance of \SBON is a \yes-instance.

$(\Leftarrow)$ Conversely, suppose that the constructed \SBON instance is a \yes-instance. Then there exists a shift vector $\sss$ such that the total cost $\sum_i \pi_i(s_i) \le b = k$, and $c_2$ becomes the winner after applying $\sss$. Let $\XX \subseteq \VV_1$ be the set of voters directly bribed (i.e., those with $s_i > 0$). Let $\XX_G = \XX \cap \VV$ denote the bribed voters corresponding to the original graph, and $\XX_I = \XX \setminus \XX_G$ the bribed isolated voters.

Define $\bar{\XX}$ to be the set of voters in $\VV$ who are either in $\XX_G$ or influenced by a voter in $\XX_G$ (i.e., $\bar{\XX} = \XX_G \cup N(\XX_G)$). Let $L = \VV \setminus \bar{\XX}$ denote the set of vertices that are neither bribed nor influenced.

Since $c_2$ wins, at least $n$ voters must now rank $c_2$ above $c_1$, implying $|\bar{\XX}| + |\XX_I| \ge n$. As $\bar{\XX}$ and $L$ partition $\VV$, we have $|\bar{\XX}| + |L| = n$, and thus $|L| \le |\XX_I|$. Because the total bribery cost is at most $k$, we have $|\XX_G| + |\XX_I| \le k$. Combining these inequalities gives
\[
|L| \le |\XX_I| \implies |\XX_G| + |L| \le k.
\]
By construction, $\XX_G \cup L$ forms a dominating set of $\GG$ (since every vertex in $L$ is explicitly chosen to cover the remaining undominated vertices in $\GG$). Therefore, $\GG$ has a dominating set of size at most $k$, and $(\GG, k)$ is a \yes-instance of \DSP.
\end{proof}
}
Note that by extension of the above theorem, the general version of \SBON is also \NPC.

\subsection{Complete graph network}
In social networks, complete graphs are interesting because they model scenarios where every individual interacts with all others. They help to capture the maximum potential for influence and information spread. 

We show that \SBON is NP-complete for complete graphs when the edge weights are arbitrary, even with just two candidates by showing a reduction from \DSP. 

Let $(\GG, k)$ be an arbitrary instance of \DSP, where $\GG = (\VV, E)$ is a graph with $|\VV| = n$ vertices. We construct, in polynomial time, an instance $(\CC, \PP, \GG_1, c_2, \Pi, b)$ of \SBON as follows.

Let the set of candidates be $\CC = \{c_1, c_2\}$, where $c_2$ is the preferred candidate of the briber. Each vertex $v_i \in \VV$ of $\GG$ corresponds to a voter $v_i$. The initial preference of each voter $v_i$ is $c_1 \succ_i c_2$, which means that all voters initially prefer $c_1$ to $c_2$. 

We add $n - 1$ additional voters $\VV_I = \{v_{n+1}, \dots, v_{2n-1}\}$, each also with preference $c_1 \succ_i c_2$. The total set of voters is therefore $\VV_1 = \VV \cup \VV_I$, and $|\VV_1|=2n-1$.  Let the overall preference profile is denoted by $\PP$. 

We now construct a complete graph $\GG_1 = (\VV_1, E_1, w)$ as follows. For every edge $\{v_i, v_j\} \in E$, set $w(\{v_i, v_j\}) = 1$. For all remaining unordered pairs $\{v_i, v_j\}$ not in $E$, add edges with weight $\frac{1}{2k}$. This ensures that $\GG_1$ is a complete weighted graph.
The family of cost functions $\Pi$ consists of:
\[
\pi_i(s_i) =
\begin{cases}
s_i, & \text{if } v_i \in \VV,\\
(k+1) \cdot s_i, & \text{if } v_i \in \VV_I.
\end{cases}
\]
That is, voters corresponding to the original vertices of $\GG$ have identity cost functions, while the additional $(n - 1)$ voters have cost functions that are prohibitively expensive to bribe. The bribery budget is set to $b = k$.
Hence, the constructed instance $(\CC, \PP, \GG_1, c_2, \Pi, b)$ of \SBON is computable in polynomial time from $(\GG, k)$. Next, we show the equivalence of these two instances in the following lemma.

\begin{lemma}[$\star$]\label{npc-complete}
The instance $(\GG, k)$ of \DSP  is a \yes-instance if and only if the constructed \SBON instance $(\CC, \PP, \GG_1, c_2, \Pi, b)$ is a \yes-instance.
\end{lemma}
\longversion{
\begin{proof}
Suppose $(\GG, k)$ is a \yes-instance of \DSP. Then there exists a dominating set $D \subseteq \VV$ such that $|D| \le k$. We bribe exactly the voters corresponding to the vertices in $D$, shifting $c_2$ to the top of their preference orders (i.e., $s_i = 1$ for $v_i \in D$ and $s_i = 0$ otherwise). Since $\pi_i(s_i) = s_i$ for these voters, the total bribery cost is at most $|D| \le b$, satisfying the budget constraint.

Because the vertices of $\GG$ form a dominating set, every vertex $v_j \in \VV \setminus D$ has at least one bribed neighbor $v_i \in D$. In $\GG_1$, the influence of each bribed vertex is weighted $w(v_i, v_j) = 1$, so every unbribed voter in $\VV$ experiences a net shift $s_j' \ge 1$ and consequently changes their top preference to $c_2$. Therefore, all $n$ voters corresponding to the original vertices of $\GG$ now prefer $c_2$ over $c_1$. Since there are $2n - 1$ voters in total, $c_2$ secures at least $n$ votes and becomes the winner. Hence, the constructed \SBON instance is a \yes-instance.

Conversely, suppose that the constructed \SBON instance is a \yes-instance. Then there exists a shift vector $\sss$ such that $\sum_i \pi_i(s_i) \le b = k$ and $c_2$ becomes the winner after the shifts are applied. Let $\XX = \{v_i \in \VV_1 \mid s_i > 0\}$ denote the set of bribed voters.

We first note that bribing any voter $v_i \in \VV_I$ (from the added set) costs at least $(k + 1)$, which exceeds the total available budget $k$ i.e. all influence must originate from voters in the original vertex set $\VV$. Moreover, to change such a voter’s effective preference via neighbors would require a total influence of at least $1$, meaning an aggregate shift contribution of at least $2k$ from other voters (since the connecting edges have weight $\frac{1}{2k}$), again exceeding the budget. Hence, no voter in $\VV_I$ can have their preference changed directly or through influence. 

Let $\XX_G = \XX \cap \VV$ denote the bribed voters corresponding to the original graph $\GG$. Because $\sum_i \pi_i(s_i) \le k$, we have $|\XX_G| \le k$. If $\XX_G$ is a dominating set of $\GG$, then $(\GG, k)$ is a \yes-instance of \DSP, as required.

For the sake of contradiction, assume that $\XX_G$ is not a dominating set. Then there exists some vertex $v_j \in \VV$ that is (not dominated), neither bribed nor adjacent to any bribed vertex in $\GG$. Since any pair non-adjacent vertices in $\GG$, connected via an an edge of weight $\frac{1}{2k}$ in $\GG_1$, the influence from all bribed vertices is at most $\frac{|\XX_G|}{2k} \le \frac{1}{2} < 1$, implying $s_j' < 1$. Thus, $v_j$ continues to prefer $c_1$ to $c_2$. Moreover, as established earlier, all $n-1$ voters in $\VV_I$ also continue to prefer $c_1$. Therefore, at most $n - 1$ voters prefer $c_2$, contradicting the assumption that the bribery scheme made $c_2$ win. Hence, $\XX_G$ must be a dominating set of $\GG$ of size at most $k$.
\end{proof}
}
From the above lemma, along with the construction, we have the following theorem.

\begin{theorem}
	\SBON is NP complete for complete graphs and 2 candidates.
\end{theorem}

In contrast to the above setting with arbitrary edge weights, if all edge weights are 1, the problem becomes polynomial-time solvable under the majority voting rule with linear cost functions and $m$ candidates.
\begin{theorem}
	\SBON is polynomial-time solvable if the underlying graph is complete and all the edge weights are 1 for the majority voting rule, and linear cost functions.
\end{theorem}

\begin{proof}
	Let $(\CC, \PP, \GG, c, \Pi, b)$ be an arbitrary instance of \SBON, under the majority voting rule. Let us say the majority rule elects a dummy candidate in the event no majority is achieved. The goal of bribery is to then make our candidate c a majority winner. The approach towards the solution is greedily solved. Notice that for all shift vectors $\sss =(s_i)_{i \in \left[ n \right]}$, the value of $s_i^\pr ~\forall i \in \left[ n \right]$ is same and equal to $\sum_{i \in [n]} s_i$. Let us call this quantity $\alpha$. \\ 
 \\
     Let us say the position vector of candidate c is $\ppp{} = \left( \succ_{i}^{\PP}\left( c \right) \right)_{i \in \left[ n  \right]}$
     where $\succ_{i}^{\PP}\left( c \right)$ specifies the position of desired candidate c in the ranking $\succ_{i}$ in preference profile $\PP$. 
     Notice that c is a Majority winner $\iff \alpha \geq$  the median of this position vector. This is because all voters will see a shift of $\alpha$, and only those voters with $\succ_{i}^{\PP}\left( c \right) \leq \alpha$ will vote for c under the plurality scoring rule. 

     Hence the minimum value of $\alpha$ needed to make c a Majority winner can be easily found using any polynomial algorithm to find median in a vector. 
    Let us say the family of linear cost functions for each voter are $\Pi = (\pi_i : [m - 1] \rightarrow \mathbb{N})_{i \in [n]}$ (where $\pi_i(x) = b_i \times x, \forall i \in [n]$), and a budget $b \in \mathbb{R}$.
    This value of $\alpha$ can be achieved with minimum possible budget by bribing only the voter with the smallest cost coefficient, at a cost of $b_i \times \alpha$, (where $b_i$ is smallest among all $i \in [n]$). This is the minimum possible cost to make c a Majority winner. If this cost is $ \leq b$, then c can be made a Majority winner. If not, c cannot be made a Majority winner.
\end{proof}

Let us now see the same problem with the plurality voting rule instead of the majority voting rule

\begin{lemma}[$\star$]\label{c_wins}
	If c is a winner under the plurality voting rule for some score $\alpha ^\pr$, then c is also a winner for all $\alpha \geq \alpha ^\pr$.
\end{lemma}
\longversion{
\begin{proof}
    As discussed previously, all voters in this problem framework shift by the same maximum amount $\alpha$. Let us say c wins for a given value of $\alpha$. This means that c has more votes than all other candidates $c ^\pr \in \CC$. We prove that this cannot change upon increasing $\alpha$. Notice that increasing the value of $\alpha$ cannot reduce the number of votes c gets, since all voters now shift c to a more preferable location. Similarly, notice that increasing $\alpha$ cannot increase the number of votes a rival of c can get, since all other candidates either stay stationary, or shift to a less preferable profile. This completes the proof.
\end{proof}}

\begin{lemma}[$\star$]\label{c_looses}
	If c is a loser under the plurality voting rule for some score $\alpha ^\pr$, then c is also a loser for all $\alpha \leq \alpha ^\pr$.
\end{lemma}

\longversion{
\begin{proof}
    Very analogous to the previous proof. Let us say c loses for a given value of $\alpha$. This means that c has less votes than some other candidate $c ^\pr \in \CC$. We prove that this cannot change upon decreasing $\alpha$. Notice that decreasing the value of $\alpha$ cannot increase the number of votes c gets, since no voter will shift c to a more preferable location. Similarly, notice that decreasing $\alpha$ cannot decrease the number of votes a rival of c can get, since all other candidates either stay stationary, or shift to a more preferable profile. 
\end{proof}}

\begin{lemma}\label{plurality_score}
	The optimal value of $\alpha$ has to be in the position vector $\ppp{} = \left( \succ_{i}^{\PP}\left( c \right) \right)_{i \in \left[ n  \right]}$ or 0. 
\end{lemma}
\longversion{
\begin{proof}
    For simplicity, let us add 0 to our position vector
    Consider a value of $\alpha$ that is not equal to any $\succ_{i}^{\PP}\left( c \right)$ for any $i \in [n]$ (or 0). Let us assume that this value of $\alpha$ is optimal (the lowest among all $\alpha$ that make c win). Now, consider the largest value of an element in our position vector that is smaller than $\alpha$. Let us call this element p. Notice that reducing the value of $\alpha$ to p cannot cause out candidate c to lose. This is because all voters with $\succ_{i}^{\PP}\left( c \right)$ less than p voted for c previously, and will continue to do so. All the other voters didn't vote for c even with the old value of $\alpha$ since their $\succ_{i}^{\PP}\left( c \right)$ was strictly greater than $\alpha$. This means that $p \geq \succ_{i}^{\PP}\left( c \right) \iff \alpha \geq \succ_{i}^{\PP}\left( c \right)$, which implies that the value p is also a valid $\alpha$ that can make c win. since p $\leq \alpha$, this contradicts the assumption that $\alpha$ is optimal.    \\\\
    Note that this proof does not work if we do not add 0 to our position vector, since $\alpha = 0$ could also be optimal (c is already the winner). 
\end{proof}
}
With these lemmas we are ready to prove the following theorem.

\begin{theorem}
	\SBON is polynomial-time solvable if the underlying graph is complete and all the edge weights are 1 for the plurality voting rule, and linear cost functions.
\end{theorem}

\begin{proof} For any given $\alpha$, we can simulate the voting procedure under the plurality rule in polynomial time and determine whether \(c\) becomes a winner. From Lemmas~\ref{c_wins} and~\ref{c_looses}, it follows that the smallest feasible $\alpha$ can be found by performing a binary search over possible $\alpha$ values.

From \Cref{plurality_score}, it follows that we only need to check $\alpha$ values from our position vector 
\(
\ppp = \left\{ \succ_i^{\mathcal{P}}(c) \;\middle|\; i \in [n] \right\} \cup \{0\},
\)
where \(\succ_i^{\mathcal{P}}(c)\) denotes the position of \(c\) in voter \(i\)'s preference ranking. Since \(|\ppp| = \OO(n)\), this yields a polynomial-size search space.

Each feasibility check for a given $\alpha$ requires generating a modified preference profile for all voters, which can be done in \(O(n)\) time. Thus, performing binary search over \ppp has time complexity
\(O(n \log n)\). Once the optimal value \(\alpha^\ast\) is identified, finding the corresponding minimum-cost bribery scheme is straightforward under linear cost functions. Specifically, if \(b_i\) denotes the cost coefficient for voter \(i\), then \(\alpha^\ast\) can be achieved at minimum cost by bribing the voter with the smallest cost coefficient, i.e.,
\(
\min_{i} \; b_i \times \alpha^\ast.
\)

If this minimum cost satisfies
\(
\min_{i} \; b_i \times \alpha^\ast \;\leq\; b,
\)
where \(b\) is the available budget, then \(c\) can be made a plurality winner; otherwise, it is impossible.

Therefore, under the stated assumptions, \SBON is solvable in polynomial time.
\end{proof}

We now turn our attention to \SBON on a transitive tournament $T=(V,A)$ with two candidates, 
identity cost functions, and uniform edge weights. 
Let the vertices be ordered by decreasing out-degree, i.e. 
$v_1,v_2,\ldots,v_n$ with $d^+(v_1)=n-1, d^+(v_2)=n-2, \ldots, d^+(v_n)=0$. 
Scanning this order from $v_1$ downwards, bribing the first vertex 
that does not already rank the preferred candidate~$1$ at the top suffices 
to make all vertices in $T$ support candidate~$1$. If all vertices already prefer candidate~$1$, no bribery is needed.
\begin{observation}
  There exists a linear time algorithm to solve \SBON on a transitive tournament graph with two candidates, uniform bribery costs, and uniform edge weights. 
\end{observation}

\subsection{Cluster graphs}
Cluster graphs are again intriguing as they model tightly connected communities in social networks. In such graphs, influence spreads easily within a cluster but less so between clusters. Studying \SBON here helps understand how shifting preferences in one community affects the overall outcome. We design pseudo-polynomial time algorithms to solve \SBON with two candidates, uniform edge weights, and linear cost functions.

\begin{theorem}[$\star$]\label{cluster-graph-pseudo-poly}
On cluster graphs with two candidates, uniform edge weights, and linear cost functions, \SBON can be solved in pseudo-polynomial time by dynamic programming in 
$\OO(r\cdot b)$, where $r$ is the number of cliques and $b$ is the budget. 
\end{theorem}
\shortversion{
\begin{proofsketch}
Assume there are two candidates i.e. $\CC= \{c_1, c_2\}$ with $c_2$ being our preferred candidate. We reduce the problem to selecting a subset of cliques to flip within budget $b$ to reach a majority for $c_2$.  
For each clique $C_i$, let $c(C_i)$ be the minimum cost to flip it and $gain(C_i)$ the number of additional supporters gained. We use dynamic programming: $DP[i][d]$ stores the maximum gain from the first $i$ cliques with budget at most $d$. Finally, we check if the total gain plus initial supporters reaches a majority.  
The DP runs in $\OO(r\cdot b)$ time where $r$ is the number of cliques in the input graph, giving a pseudo-polynomial algorithm.
\end{proofsketch}
}
\longversion{
\begin{proof}
Let $\CC = \{c_1,c_2\}$ with $c_2$ being our preferred candidate. Let $\GG=(V,E)$ be a cluster graph consisting of disjoint cliques 
$C_1,\dots,C_r$. 
Each voter $v \in V$ has an arbitrary cost $c(v)$ to shift the preferred candidate~$c_2$ to the top. 
For each clique $C_i$, define
\[
\begin{aligned}
c(C_i) \;&=\; \min_{v \in C_i} c(v), \\[4pt]
gain(C_i) \;&=\; |C_i| \;-\; 
   \#\{\,v \in C_i \mid v \text{ already ranks } c_2 \text{ on top}\,\}.
\end{aligned}
\]

That is, $c(C_i)$ is the minimum cost to flip $C_i$, and $gain(C_i)$ is the number of additional supporters gained if $C_i$ is flipped.

The problem reduces to selecting a subset of cliques whose total cost does not exceed $b$ and whose total gain, together with the initial supporters, reaches the majority threshold $\lceil (n+1)/2 \rceil$.

We construct the following dynamic programming table:
\[
DP[i][d] \;=\; 
\max \Bigl\{ \sum_{j \in S} gain(C_j) \;\Bigm|\; 
S \subseteq \{1,\ldots,i\}, \;
\sum_{j \in S} c(C_j) \leq d \Bigr\}.
\]
That is, $DP[i][d]$ stores the maximum number of additional supporters obtainable from the first $i$ cliques using budget at most $d$.

The recurrence is
\(
DP[i][d] \;=\; 
\max \bigl(
DP[i-1][d], \;
DP[i-1][d - c(C_i)] + gain(C_i)
\bigr),
\)
whenever $d \geq c(C_i)$, and otherwise
\(
DP[i][d] \;=\; DP[i-1][b].
\)

The base case is $DP[0][d]=0$ for all $d \geq 0$.  
At the end, let $l$ be the number of initial supporters of candidate~$c_2$.  
We check whether
\(
\max_{d\leq b} DP[r][d] + l \;\;\geq\;\; \lceil (n+1)/2 \rceil.
\)
If yes, then the bribery succeeds within budget $b$; otherwise, it is impossible.

The DP table has size $O(r \cdot b)$ and each entry can be computed in constant time from previous entries. Thus, the algorithm runs in $O(r\cdot b)$ time, which is pseudo-polynomial in the budget $b$. 
\end{proof}
}
\subsection{Path graph network}

We consider a uni-directed path of \(n\) voters with unit edge weights and identity cost functions (i.e., \(\pi_i(s)=s\) for all \(i\)).  Concretely, for \(i=1,\dots,n-1\) the path contains the arc \((i+1,i)\) with weight \(1\) (influence flows leftwards). Consider that the shift vector  $(s_1^\pr, \dots, s_n^\pr) \in \mathbb{N}_0^n$. It now follows from our problem definition that $s_i^\pr = s_i + s_{i+1}$. Under plurality, for a voter $i$ to vote for c,  $s_i^\pr \geq \succ_{i}^{\PP}\left( c \right)$.

\begin{lemma}[$\star$]\label{left-sgift tight}
	There exists an optimal solution where $s_i \leq \succ_{i}^{\PP}\left( c \right) \Rightarrow s_i + s_{i-1} = \succ_{i-1}^{\PP}\left( c \right)$  
\end{lemma}
\longversion{
\begin{proof}
    First consider the case where $s_i + s_{i-1}$ is greater than $\succ_{i-1}^{\PP}\left( c \right)$. Now let us say there exists an optimal solution with $s_i > \ \succ_{i}^{\PP}\left( c \right)$. We can decrease the value of $s_i$ until $s_i = max(\succ_{i}^{\PP}\left( c \right), \succ_{i-1}^{\PP}\left( c \right) - s_{i - 1})$ without changing optimality. This is because we do not decrement $s_i$ past the point where either voter changes their vote. We have now obtained a new optimal solution that satisfies the above property.
    
    Now consider the case where $s_i + s_{i-1} < \succ_{i-1}^{\PP}\left( c \right)$. Once again let us say there exists an optimal solution with $s_i > \succ_{i}^{\PP}\left( c \right)$. We can decrease the value of $s_i$ to $\succ_{i}^{\PP}\left( c \right)$ without changing optimality. Since the $i^{th}$ voter will continue voting for candidate c and the $i - 1 ^{th}$ voter will continue to not vote for c. the value of ${s_i}$ does not impact any other voters. This completes the proof of the lemma.
\end{proof}
}

\begin{lemma}[$\star$]\label{right-shift-tight}
	There exists an optimal solution where $s_{i+1} > 0 \Rightarrow s_i + s_{i+1} = \succ_{i}^{\PP}\left( c \right)$.
\end{lemma}
\longversion{
\begin{proof}
	Suppose there exists an optimal shift vector $\mathbf{s}$ such that $s_{i+1} > 0$ but 
	\(
		s_i + s_{i+1} \neq \succ_{i}^{\PP}(c).
	\)
	We consider two cases.

	\paragraph{Case 1:} $s_i + s_{i+1} > \succ_{i}^{\PP}(c)$.  
	In this case, we can reduce $s_{i+1}$ to
	\(		s_{i+1}' := \min\left(0, \succ_{i}^{\PP}(c) - s_i\right),
	\)
	and simultaneously increase $s_{i+2}$ by the amount
	\(
		\delta := s_{i+1} - s_{i+1}'.
	\)
	This transformation preserves the condition that voter $i$ votes for candidate $c$, since
	\(
		s_i + s_{i+1}' = \min\left(s_i, \succ_{i}^{\PP}(c)\right) \ge \succ_{i}^{\PP}(c).
	\)
	Similarly, voter $i+1$’s decision remains unchanged since the total influence on voter $i+1$ depends on $s_i + s_{i+1}$, which is unchanged by this transformation.  
	Furthermore, subsequent voters can only be better off since increasing $s_{i+2}$ increases their total shift.  
	Since all voters have identity cost functions and the transformation only reallocates shift values without increasing their sum, the total cost is unaffected.  
	Thus, we obtain another optimal solution satisfying 
	\(		s_i + s_{i+1} = \succ_{i}^{\PP}(c).
	\)

	\paragraph{Case 2:} $s_i + s_{i+1} < \succ_{i}^{\PP}(c)$.  
	In this case, we can reduce $s_{i+1}$ to $0$ and increase $s_{i+2}$ by $s_{i+1}$. This modification preserves voter $i$’s decision, since $s_i + s_{i+1} < \succ_{i}^{\PP}(c)$ implies voter $i$ would not change their vote regardless of the decrease.  
	Voter $i+1$’s decision is unaffected because the total influence they experience remains unchanged.  
	Subsequent voters again can only be better off due to the increased $s_{i+2}$.  
	The cost remains unchanged due to the identity of cost functions and simple reallocation of shift amounts.  
	Thus, this transformation yields another optimal solution satisfying the property.

	In both cases, we have shown that an optimal solution exists in which
	\(
		s_{i+1} > 0 \ \Rightarrow \ s_i + s_{i+1} = \succ_{i}^{\PP}(c).
	\)
\end{proof}
}
\begin{lemma}[$\star$]
	There exists an optimal solution where $s_0$ is 0.  
\end{lemma}
\longversion{
\begin{proof}
    Similar logic to above, transfer $s_0$ to $s_1$. This does not affect $\succ_{0}^{\PP}\left( c \right)$, and can only positively affect subsequent voters.
\end{proof}
}
From the preceding lemmas, it follows that for every value of $s_{i+1}$, it suffices to ensure that the condition 
\(
	s_i + s_{i+1} = \succ_{i}^{\PP}(c)
\)
is either satisfied exactly or disregarded entirely by setting $s_{i+1} = 0$.  

Since all cost functions are identical and linear, i.e., $\pi_i(s_i) = s_i$ for all $i \in [n]$, the objective reduces to minimizing the total cost
\(
	\sum_{i \in [n]} s_i.
\)

For the majority voting rule, this observation naturally leads to a dynamic programming formulation. Specifically, the problem reduces to deciding, for each voter, whether to satisfy their shift requirement or to ignore it completely, subject to the majority constraint. This decision structure enables a dynamic programming approach, described as follows.

Consider a directed path of $n$ voters $v_n \to v_{n-1} \to \dots \to v_1$, where influence flows leftwards along the path. Each voter $v_i$ can be shifted directly by the briber and influenced by the voter to the right ($v_{i+1}$). Let $\tau_i$ denote the minimum effective shift required for voter $v_i$ to vote for the preferred candidate $c$, i.e., the initial rank of $c$ in $\succ_i^\PP$. Let $\pi_i$ denote the cost of applying one unit of shift to voter $v_i$ (identity cost or general).

\paragraph{Effective shift.} For $i = 1,\dots,n-1$, the effective shift at voter $i$ is
\[
s_i' = s_i + s_{i+1}, \quad s_{n+1} := 0,
\]
where $s_i$ is the direct shift applied to voter $i$ and $s_{i+1}$ is the influence propagated from voter $i+1$. Voter $i$ is \emph{convinced} if $s_i' \ge \tau_i$.

We define the DP state
\(
dp[i,s_{i},b^\pr] = \text{max. number of voters convinced among } \{v_i, v_{i-1}, \dots, v_1\}
\)
when processing from right to left, with:
\begin{itemize}
    \item $i$ = current voter being processed,
    \item $s_{i}$ = shift applied to voter $v_{i}$,
    \item $b^\pr$ = budget for direct shifts.
\end{itemize}

\paragraph{Base case.} For the rightmost voter $v_n$:
\[
dp[n,s_n,b^*] = 
\begin{cases}
1 &\text{if } \tau_n = 0\\
1 & \text{if } s_n + s_{n+1} \ge \tau_n \text{ and } s_n \cdot \pi_n \le b^*, \\
0 & \text{otherwise},
\end{cases}
\].

\paragraph{Transition.} For $i = n-1,\dots,1$ :
\[
\begin{aligned}
dp[i,s_i,b^\pr] = \max \Bigl\{\, 
& dp[i\!+\!1,s_{i+1},b^\pr - s_i \pi_i] 
   + \mathbf{1}_{\,s_i + s_{i+1} \ge \tau_i,\; s_i\pi_i\le b`},\\
& dp[i\!+\!1,s_{i+1},b^\pr]
\Bigr\}.
\end{aligned}
\]
Here:
\begin{itemize}
    \item $s_i$ = direct shift applied to voter $v_i$,
    \item $s_i + s_{i+1}$ = effective shift accounting for influence from $v_{i+1}$,
    \item the DP value increases by $1$ if voter $i$ is convinced,
    \item for the next step (leftwards), $s_i$ becomes the propagated shift $s_{i}$ to $v_{i-1}$.
\end{itemize}

Let b be the total budget. The maximum number of voters that can be convinced is
\(
\max \{\min_{0\leq s_1 \leq \tau_1}dp[1,s_1,b]\}
\)
and check if this attains at least $\lceil{\frac{n+1}{2}\rceil}$ or not.

\paragraph{Complexity.} Let $\tau_{\max} = \max_i \tau_i$, which can be upper bounded by $n$. The total no. of DP states= $\OO(n \cdot (\tau_{\max}+1) \cdot (b+1))$. Each state considers $O(b / c_i)$ feasible shifts $s_i$, giving $\OO(n^2 \cdot b)$ with identity costs.  
- Using standard optimization, the time and space complexity remains polynomial in $n$ and b.


\begin{theorem}
There exists a dynamic programming algorithm that computes an optimal shift vector \(\sss^\ast = (s_1^\ast, \dots, s_n^\ast)\) for \SBON on a path graph with linear cost functions and the majority voting rule, running in time  \(\OO(n^2 \cdot b)\).
\end{theorem}

\subsection{Bipartite Graphs}
In this section, we show that \SBON is \NPC even for bipartite graphs. We give a polynomial time reduction from \textsc{Set Cover}, 
which is known to be \NPC.

An instance of \textsc{Set Cover} consists of a universe $\mathcal{U}$ of $n$ elements, 
a family $\alpha=\{S_1,\ldots,S_m\}$ of $m$ subsets of $\mathcal{U}$, 
and an integer $k$. The task is to determine whether there exists a subfamily $\{S_{i_1},\ldots,S_{i_k}\}$ that covers all of $\mathcal{U}$.

Given an instance $(\mathcal{U},\alpha,k)$ of \textsc{Set Cover}, construct an undirected bipartite graph $G'=(L\cup R,E')$ as follows.  
For each element $u\in \mathcal{U}$, create a vertex in $L_{\mathrm{orig}}$.  
For each set $S_i\in \alpha$, create a vertex $r_i\in R_{\mathrm{orig}}$, and add edges $\{r_i,u\}$ for all $u\in S_i$.  
Next, add the following gadgets. Introduce $m-k+1$ new vertices $L_{\mathrm{new}}$ to $L$, each adjacent to every vertex of $R_{\mathrm{orig}}$ (i.e., complete bipartite connections).  
Also add $n+k-1$ isolated vertices $R_{\mathrm{new}}$ to $R$.  
Thus, the bipartition is $L=L_{\mathrm{orig}}\cup L_{\mathrm{new}}$ and $R=R_{\mathrm{orig}}\cup R_{\mathrm{new}}$. The total number of vertices is $2(n+m)$.

We now construct an equivalent \SBON instance. There are two candidates, denoted by $c_1$ and $c_2$, where $c_2$ is the preferred candidate. Initially, every voter has a preference order $c_1 \succ c_2$. Bribing a voter costs~$1$. The budget is set to $b=k$. The winning condition is that candidate~$c_2$ must obtain at least $t= n+m+1$ supporters.

\begin{lemma}[$\star$]\label{sc_equi}
There exists a set cover of size at most $k$ if and only if there exists a bribery strategy for \SBON of cost at most $k$.
\end{lemma}
\longversion{
\begin{proof}
Suppose first that the given \textsc{Set Cover} instance is a \yes-instance, and let $\{S_{i_1},\ldots,S_{i_k}\}$ be a cover of size $k$. Bribe the corresponding vertices $\{r_{i_1},\ldots,r_{i_k}\}\subseteq R_{\mathrm{orig}}$. Then all $n$ vertices in $L_{\mathrm{orig}}$ are influenced, since the chosen subfamily covers $\mathcal{U}$. Moreover, all $m-k+1$ vertices in $L_{\mathrm{new}}$ are influenced, as they are adjacent to every vertex in $R_{\mathrm{orig}}$. Finally, the $k$ bribed vertices in $R_{\mathrm{orig}}$ themselves support candidate~$c_2$. Therefore, the total number of supporters of $c_2$ is $n+(m-k+1)+k = n+m+1$, meeting the majority threshold. Hence the constructed \SBON instance is a \yes-instance.

Conversely, suppose that the constructed \SBON instance is a \yes-instance. Then there exists a bribery set $B$ with $|B|\leq k$ such that at least $n+m+1$ voters support candidate~$c_2$ after bribery and propagation. Partition $B$ as
\[
B=(L_1\cup L_2)\cup(R_1\cup R_2),
\]
where $L_1\subseteq L_{\mathrm{orig}}$, $L_2\subseteq L_{\mathrm{new}}$, 
$R_1\subseteq R_{\mathrm{orig}}$, and $R_2\subseteq R_{\mathrm{new}}$.

\begin{claim}
There exists a solution with $B\subseteq R_{\mathrm{orig}}$.
\end{claim}

\begin{proof}
We first argue that an optimal bribery strategy of cost at most $k$ may, without loss of generality, be assumed to bribe only vertices in $R_{\mathrm{orig}}$. Indeed, bribing a vertex of $L_1$ influences only that vertex, whereas bribing one of its neighbors in $R_{\mathrm{orig}}$ influences not only that vertex of $L_1$ but possibly additional vertices as well. Hence any bribery that targets a vertex of $L_1$ can be replaced by bribing one of its neighbors in $R_{\mathrm{orig}}$ without decreasing the number of supporters. Similarly, bribing a vertex of $L_2$ is redundant, since every vertex of $L_2$ is already influenced whenever any neighbor in $R_{\mathrm{orig}}$ is bribed. Finally, bribing a vertex of $R_2$ yields only that single supporter, since these vertices are isolated, and replacing such a bribery with one targeting a vertex of $R_{\mathrm{orig}}$ never decreases the total number of influenced vertices. 
\end{proof}

Next, consider the requirement of reaching $n+m+1$ supporters. Since all $m-k+1$ vertices of $L_{\mathrm{new}}$ must be influenced, and each such vertex has neighbors only in $R_{\mathrm{orig}}$, this necessitates that $B$ contain at least one vertex of $R_{\mathrm{orig}}$. Moreover, all $n$ vertices of $L_{\mathrm{orig}}$ must be influenced as well, for otherwise the total number of supporters is strictly less than $n+m+1$. Consequently, every vertex of $L_{\mathrm{orig}}$ is adjacent to some bribed vertex in $R_{\mathrm{orig}}$, which means that $B$ dominates all of $L_{\mathrm{orig}}$ and satisfies $|B|\leq k$. By construction of the reduction, each vertex of $R_{\mathrm{orig}}$ corresponds to a set in the family of the original set cover instance, and domination of $L_{\mathrm{orig}}$ by $B$ corresponds precisely to covering the universe $\mathcal{U}$. Therefore, the bribery strategy $B$ induces a set cover of $\mathcal{U}$ of size at most $k$.
\end{proof}
}
With this lemma, we arrive at the following theorem.

\begin{theorem}
    \SBON for bipartite graphs is \NPC even for two candidates, identity cost functions and uniform weights.
\end{theorem}
\section{Results: Parameterized Algorithms}
We discuss the plausibly of parametrized algorithms for \SBON. While the general problem is intractable, certain structural properties of the underlying social network can make the problem efficiently solvable. In real-world networks, parameters such as treewidth or feedback vertex set size are often small, even when the overall graph is large. Exploiting these parameters allows us to design fixed-parameter tractable (\FPT) algorithms that efficiently handle such instances. In this section, we explore how the problem behaves under these natural graph parameters and election parameters and present corresponding \FPT algorithms or establish W-Hardness that leverage the bounded structural complexity of the network.
\subsection{Bounded Treewidth as a Parameter}
We consider a nice tree decomposition $(\mathbb{T},\mathcal{X})$ of the underlying graph. 
Let $(\CC, \PP, \GG, c, \Pi, b)$ be an input instance of \SBON 
such that $\tw = \tw(\GG)$. 
We define a function $\ell : V_{\mathbb{T}} \rightarrow \mathbb{N}$ as follows. 
For the root $r$ of $\mathbb{T}$, we set $\ell(r) = 0$, and for a node $t \in V_{\mathbb{T}}$, 
we let $\ell(t) = \sf{dist}_{\mathbb{T}}(t,r)$, where $\sf{dist}_{\mathbb{T}}$ 
denotes the distance in the decomposition tree. 
Thus, $\ell(r)=0$, and in general the values of $\ell$ range between $0$ and $L$. 

For a node $t \in V_{\mathbb{T}}$, we denote by $V_t$ the set of vertices 
contained in the bags in the subtree of $\mathbb{T}$ rooted at $t$, 
and let $\GG_t = \GG[V_t]$ denote the subgraph induced by $V_t$. 

Now, we describe a dynamic programming algorithm over $(\mathbb{T},\mathcal{X})$. 
We define the following states.

\paragraph{DP State.}
For each node $t \in V_{\mathbb{T}}$ with bag $X_t$, a subset $S \subseteq X_t$ of bribed vertices, 
a subset $I \subseteq X_t$ of influenced vertices, and a budget value $0 \leq b' \leq b$, 
we define the dynamic programming entry
\(
DP[t,S,I,b'] \;=\; 
\text{maximum number of influenced vertices in the subtree of $t$}.
\)

subject to:
\begin{itemize}
    \item exactly the vertices in $S$ are bribed in $X_t$,
    \item exactly the vertices in $I$ are influenced in $X_t$,
    \item the total bribery cost inside the subtree of $t$ is $b'$.
\end{itemize}
If no such configuration exists, we set $DP[t,S,I,b'] = -\infty$.

\paragraph{Transitions.}
We give the recurrence rules for the four node types in a nice tree decomposition.

\paragraph{Leaf.} 
If $X_t = \varnothing$, then 
\(
DP[t,\varnothing,\varnothing,0] = 0,
\)
and all other entries are $-\infty$.

\paragraph{Introduce node.} 
Suppose $t$ introduces a vertex $v$, so $X_t = X_{t'} \cup \{v\}$. 
For a state $(S,I,b')$ at $t$, let $S' = S \setminus \{v\}$ 
and consider a child state $(S',I',b'')$ at $t'$.
We require $I$ to be consistent with $(S,I')$:
\begin{itemize}
    \item If $v \in S$, then $v \in I$ and we pay cost $c(v)$. 
    Moreover, every neighbor $u \in N(v)\cap X_{t'}$ must belong to $I$ 
    (since bribing $v$ influences them). We look at the child node $t'$ with $S' = S \setminus \{v\}$ and $I'= I \setminus (N(v)\cap X_{t'})$.
    In this case,
    \begin{equation*}
\begin{split}
DP[t,S,I,b'] \;=\;&\; DP[t',S',I',b'-c(v)] + 1 \\
&+ \big|\{ u \in N(v)\cap X_{t'} : u \notin I'\}\big|.
\end{split}
\end{equation*}
    \item If $v \notin S$, then $v \in I$ only if some bribed neighbor in $V_t$ influences $v$; 
    otherwise $v \notin I$. In either case we require $I \cap X_{t'} = I'$, and
    \[
    DP[t,S,I,b'] \;=\; 
    DP[t',S',I',b'] 
    + \mathbf{1}_{\{\exists u \in N(v)\cap V_t : u \in S'\}}.
    \]
\end{itemize}

\paragraph{Forget node.} 
Suppose $t$ forgets a vertex $v$, so $X_t = X_{t'} \setminus \{v\}$. 
Since vertices are already counted when they become influenced, 
no extra contribution is added here. 
We simply set
\[
\begin{aligned}
DP[t,S,I,b'] = \max \big\{ & DP[t', S \cup \{v\}, I \cup \{v\}, b'], \\
                            & DP[t', S, I \cup \{v\}, b'], \\
                            & DP[t', S, I, b'] \big\}.
\end{aligned}
\]

\paragraph{Join node.} 
Suppose $t$ has two children $t_1,t_2$ with the same bag $X_t$. 
Both children must agree on $S$ and $I$ for the bag. 
Since bag vertices are counted in both subtrees, 
we subtract $|I|$ once to avoid double counting:
\[
DP[t,S,I,b'] \;=\; 
\max_{b_1+b_2=b'} 
\Bigl( DP[t_1,S,I,b_1] + DP[t_2,S,I,b_2] - |I| \Bigr).
\]

\paragraph{Root.} 
At the root $r$, the total number of influenced vertices is simply
\(
\max_{S,I,\,b'\leq b} \;DP[r,S,I,b'].
\)

\begin{lemma}[$\star$]\label{treewidth-correctness}   
Let $\mathbb{T}$ be the nice tree decomposition used by the DP and let $\tw$ denote the treewidth (so every bag has size at most $\tw+1$). Suppose the DP table \(
DP[t,S,I,b']\)
has been computed for every node $t$ of $\mathbb{T}$, every $S,I\subseteq X_t$, and every $0\le b'\le b$.  
Then one can recover an explicit optimal bribery (equivalently, the binary shift vector $s\in\{0,1\}^n$) that achieves
\(
\max_{S,I,\,b'\le b} DP[r,S,I,b']
\)
\end{lemma}

The correctness of this lemma allows us to prove the following result.

\longversion{
\begin{proof} We maintain the following invariant: for every node $t$ of the nice tree decomposition with bag $X_t$, for every $S\subseteq X_t$, $I\subseteq X_t$, and $0\le b'\le b$, the entry
\(
DP[t,S,I,b']
\)
equals the maximum number of vertices in $V_t$ that can be influenced by any bribery configuration restricted to the subtree rooted at $t$ that
\begin{enumerate}
  \item bribed exactly the vertices $S$ inside the bag $X_t$,
  \item results in exactly the set $I$ of influenced vertices inside the bag $X_t$, and
  \item spends total bribery cost exactly $b'$ inside $V_t$.
\end{enumerate}
If no such configuration exists, $DP[t,S,I,b']=-\infty$.

We prove the invariant by induction on the height (distance from the root) of node $t$.

\paragraph{\textbf{Base (leaf)}.}
If $t$ is a leaf then $X_t=\varnothing$ and $V_t=\varnothing$. The only feasible configuration uses zero budget and influences zero vertices; hence
\(
DP[t,\varnothing,\varnothing,0]=0,
\)
and all other entries are infeasible ($-\infty$). This matches the invariant.

\textbf{Inductive step.}
Assume the invariant holds for all children of $t$. We show each node-type transition computes the correct maximum and therefore the invariant holds at $t$.
\begin{enumerate}
    \item Introduce node. Suppose $t$ introduces vertex $v$ and its child is $t'$ with $X_t = X_{t'}\cup\{v\}$. Fix $(S,I,b')$ for $t$. Any feasible completion of this partial assignment in $V_t$ induces a feasible completion in $V_{t'}$ with bag assignment $(S',I')$ where $S' = S\setminus\{v\}$ and $I' = I\cap X_{t'}$. There are exactly three mutually exclusive possibilities for $v$ relative to the bribery/influence decisions:
\begin{enumerate}
  \item $v\in S$ (we bribed $v$). Then $v$ must be in $I$, we pay cost $c(v)$, and bribing $v$ may additionally influence neighbours of $v$ that lie in $X_{t'}$. Every feasible completion in $V_t$ corresponds to a feasible completion in $V_{t'}$ with $S',I'$ where neighbours already counted are removed as in the recurrence. By the inductive hypothesis the child DP value yields the maximum number of influenced vertices in $V_{t'}$ consistent with that child assignment; adding $v$ and the newly influenced neighbours yields exactly the total in $V_t$. The recurrence
   \[
    DP[t,S,I,b'] \;=\; 
    DP[t',S',I',b'-c(v)] + 1 
    + \big|\{ u \in N(v)\cap X_{t'} : u \notin I'\}\big|.
    \]
  therefore computes the maximum over all completions in this case (and is $-\infty$ when $b'-c(v)<0$).
  \item $v\notin S$ but $v\in I$ (not bribed but already influenced). Then some bribed neighbour in $V_t$ must influence $v$. Any such completion restricts to a child completion with $(S',I',b')$ and the validity of $v\in I$ is exactly captured by the indicator $\mathbf{1}_{\{\exists u\in N(v)\cap X_{t'}:u\in S'\}}$ (or more generally, check for bribed neighbours in the processed part). By the inductive hypothesis, $DP[t',S',I',b']$ is the maximum achievable in the child; adding the one for $v$ when allowed gives the correct value.
  \item $v\notin S$ and $v\notin I$. Then $v$ contributes nothing and feasible completions correspond one-to-one with child completions having $(S',I',b')$, so
  \[
  DP[t,S,I,b'] = DP[t',S',I',b']
  \]
  is correct.
\end{enumerate}
Since these three cases are exhaustive and mutually exclusive, taking the maximum over them yields the correct optimal value for $DP[t,S,I,b']$.

\item Forget node. Suppose $t$ forgets vertex $v$ so $X_t = X_{t'}\setminus\{v\}$. Fix $(S,I,b')$ at $t$. Any feasible completion in $V_t$ arises from some feasible completion in $V_{t'}$ where $v$ had one of the three local statuses in the child bag: (i) bribed and influenced, (ii) not bribed but influenced, or (iii) not bribed and not influenced. These correspond precisely to the three child states
\[
(S\cup\{v\},I\cup\{v\},b'),\quad (S,I\cup\{v\},b'),\quad (S,I,b'),
\]
respectively. By the inductive hypothesis each child DP value is the maximum over completions in the child consistent with that child state, and since nothing additional is gained at the forget operation (vertices are already counted when influenced), the maximum of these three child values equals the maximum over all completions in $V_t$ consistent with $(S,I,b')$. Hence the recurrence
\[
DP[t,S,I,b'] \;=\;
\max\left\{
\begin{aligned}
   &DP[t',S\cup\{v\},I\cup\{v\},b'], \\
   &DP[t',S,I\cup\{v\},b'], \\
   &DP[t',S,I,b']
\end{aligned}
\right\}.
\]
is correct.

\item Join node. Suppose $t$ has two children $t_1,t_2$ with $X_{t_1}=X_{t_2}=X_t$. Fix $(S,I,b')$ at $t$. Any feasible completion in $V_t$ decomposes into completions in the two child subtrees whose bag-assignments agree with $(S,I)$ and whose budgets sum to $b'$. Each child DP (by the inductive hypothesis) equals the maximum number of influenced vertices in its subtree consistent with the child assignment. However, vertices in the bag $X_t$ (in particular the influenced set $I$) are counted in both child DP values, so we must subtract $|I|$ once to avoid double counting. Maximizing over all budget splits $b_1+b_2=b'$ yields:
\[
DP[t,S,I,b'] = \max_{b_1+b_2=b'} \big( DP[t_1,S,I,b_1] + DP[t_2,S,I,b_2] - |I| \big),
\]
which therefore gives the correct maximum for $V_t$.

\end{enumerate}
\medskip\noindent\textbf{Root.} Let $r$ be the root with bag $X_r$. The whole graph equals $V_r$. By the invariant, for each $(S,I,b')$ the value $DP[r,S,I,b']$ equals the maximum number of influenced vertices achievable in the whole graph under those bag constraints; hence taking
\[
\max_{S,I,\,b'\le b} DP[r,S,I,b']
\]
gives the optimal number of influenced vertices for the input instance. This completes the induction and correctness proof. Now it remains to check whether \(
\max_{S,I,\,b'\le b} DP[r,S,I,b'] \geq \frac{n}{2} + 1
\)
or not.

\paragraph{Reconstruction.}
Since we only have two candidates, every bribery corresponds to 
shifting our preferred candidate one position upwards in the vote. 
Thus, the shift vector is binary: for each voter $v$, 
\[
s_v \;=\;
\begin{cases}
1 & \text{if $v$ was bribed}, \\
0 & \text{otherwise}.
\end{cases}
\]
Using standard Dynamic programming techniques, we can track which vertices are bribed.
\end{proof}
}
\begin{theorem}[$\star$]\label{twproof}
For elections with two candidates, identity cost functions and uniform edge weights, 
the \SBON problem can be solved in 
\(
\OO\!\bigl(4^{\tw}\cdot n^{\OO(1)} \cdot \kappa\bigr)
\)
time, where $\tw$ is the treewidth of the underlying graph 
and 
\(
\kappa \;=\; 
\max\Bigl\{0,\; \bigl\lceil\tfrac{n+1}{2}\bigr\rceil - \ell\Bigr\},
\)
with $\ell$ denoting the number of voters already voting for the 
preferred candidate.
\end{theorem}
\longversion{
\begin{proof}
Let $\tw := |X_t| \le \operatorname{tw}(\GG)+1$ denote the (maximum) bag size. For each node $t$ we store an entry for every triple $(S,I,b')$ where $S\subseteq X_t$, $I\subseteq X_t$, and $b'\in\{0,1,\dots,b\}$. Hence the number of states per node is at most
\(2^{\tw}\cdot 2^{\tw}\cdot (b+1) = 4^{\tw}(b+1).
\)
Since the treewidth of the input graph $\GG$ is at most $\tw$, 
it is possible to construct a data structure in time $\tw^{\OO(1)} \cdot n$ 
that allows performing adjacency queries in time $\OO(\tw)$. 
As we may assume without loss of generality that the number of nodes in the 
tree decomposition is $\OO(\tw \cdot n)$, and there are $n$ choices for the 
introduced vertex $v$, the running time of the algorithm is
\(
\OO\!\left(4^{\tw}\cdot n^{\OO(1)} \cdot b\right).
\)

Since we only have two candidates, every bribery corresponds to a unit shift, 
and hence each bribed voter contributes cost~1. Therefore $b$ can be bounded 
as follows:
\(
b \;\leq\; n-\ell,
\)
where $\ell$ is the number of voters already voting for the preferred 
candidate. Moreover, if the goal is to ensure victory of the preferred 
candidate, it suffices to consider
\[
b \;\leq\; \kappa \;=\; 
\max\Bigl\{0,\; \bigl\lceil\tfrac{n}{2}\bigr\rceil+1 - \ell\Bigr\},
\]
since $\kappa$ is exactly the number of additional votes required to reach a 
strict majority. 

Thus, the overall running time of our algorithm is
\(
\OO\!\bigl(4^{\tw}\cdot n^{\OO(1)} \cdot \kappa\bigr),
\)
where $\kappa \leq \lceil (n + 1)/2 \rceil$.
\end{proof}
}
Using the above result, we can show that \SBON parameterized by the size of a Feedback Vertex Set is \FPT. The intuition follows from our earlier result establishing \FPT with respect to treewidth. Since removing the feedback vertex set say $X$ yields a forest, we can exhaustively enumerate all possible subsets of $X$ (there are $2^{|X|}$ possibilities), and for each such configuration, solve the resulting instance in polynomial time using the dynamic programming algorithm designed for bounded-treewidth graphs. As the number of subsets depends only on $k=|X|$, the overall running time is $\OO(2^k)\cdot n^{\OO(1)}$. This leads us to the following theorem:

\begin{theorem}\label{fvs}
For elections with two candidates and unit shift cost, 
the \SBON problem can be solved in 
\(
\OO\!\bigl(2^{k}\cdot n^{\OO(1)} \cdot \kappa\bigr)
\)
time, where $k$ is the size of a minimum feedback vertex set of the underlying graph 
and 
\(
\kappa \;=\; 
\max\Bigl\{0,\; \bigl\lceil\tfrac{n+1}{2}\bigr\rceil - \ell\Bigr\},
\)
with $\ell$ denoting the number of voters already voting for the 
preferred candidate.
\end{theorem}

\subsection{Cluster Vertex Deletion Number as a parameter}
The pseudo polynomial time algorithm for cluster graphs in \Cref{cluster-graph-pseudo-poly} motivates us to exploit the \FPT algorithm for graphs with small cluster vertex deletion number, denoted by $k$.
\begin{theorem}[$\star$]\label{cvd-fpt}
\SBON can be solved in pseudo polynomial time if the cluster vertex deletion number of the underlying graph is a constant. Consequently, \SBON is pseudo fixed-parameter tractable (\FPT) when parameterized by the cluster vertex deletion number for two candidate, uniform edge weights and linear cost functions.
\end{theorem}

\longversion{
\begin{proof}
Let $\GG=(V,E)$ be the underlying graph of the given instance of \SBON. 
Suppose that $X \subseteq V$ is a cluster vertex deletion set of size $k$, i.e.\ $\GG - X$ is a cluster graph. 
It is known that such a set $X$ can be computed in FPT time with respect to~$k$`\cite{boral2016fast}.

Since $k$ is constant, we can exhaustively try all subsets $Y \subseteq X$ that may be part of the bribery scheme. 
For each choice of $Y$, we compute the cost of bribing $Y$ and the gain obtained in terms of additional supporters of the preferred candidate. 
If the cost of bribing $Y$ already exceeds the budget, we discard this choice.

Now consider the residual instance on $\GG - X$, which is a cluster graph. 
The effect of bribing vertices in $Y$ is that some vertices of $\GG - X$ may already support the preferred candidate (either directly if bribed, or indirectly via propagation). 
Thus, the residual problem on $\GG - X$ is precisely an instance of \SBON on a cluster graph with updated initial supporters and reduced budget. 
This can be solved in pseudo-polynomial time $O(r \cdot b)$ using the dynamic programming algorithm from the previous theorem, where $r$ is the number of cliques in $\GG - X$ and $b$ is the remaining budget.

Since the number of subsets $Y \subseteq X$ is $2^k$ and $k$ is a constant, this branching contributes only a constant factor. 
Hence the overall running time is $O(2^k \cdot poly(r \cdot b)$. 
For constant $k$, this is polynomial in the input size. 
More generally, this shows that the problem is pseudo fixed-parameter tractable when parameterized by $k$.
\end{proof}
}
\begin{corollary}
\SBON is \FPT  when parameterized by the cluster vertex deletion number for two candidate, uniform edge weights and identity cost functions.      
\end{corollary}

\subsection{Parameterized by number of affected voters}

We now turn our attention to election parameters. A very natural parameter is the number of affected voters which considers both bribed and influenced voters required to achieve absolute majority.

We first show that our problem translates to the \ktDSP which we define below:

\begin{definition}[\ktDSP] Given a undirected graph $G=(V,E)$ and integers $k$ and $t$, 
if there exists 
a set $D \subseteq V$ with $|D| \leq k$ such that at least $t$ vertices 
of $G$ are \emph{dominated} by $D$.
\end{definition}

Notice that the above problem is \NPC even for $t = \left\lceil \frac{|V| + 1}{2} \right\rceil$.  
This follows by a simple reduction from \DSP: given $(\GG=(V,E),k)$ with $|V|=n$,  
we construct $G'$ by adding $n-1$ isolated vertices, set $k'=k$, and $t' = n$.  
Any dominating set of $G$ of size at most $k$ dominates all $n$ original vertices in $G'$,  
thus satisfying the threshold $t'$. Conversely, any solution to this \ktDSP instance must dominate all original vertices,  
implying a dominating set of $G$.

\begin{observation}\label{dsp-hard}
\ktDSP is \NPC even when $t = \left\lceil \frac{|V| + 1}{2} \right\rceil$. 
\end{observation}

Using \Cref{dsp-hard}, we prove the following theorem.

\begin{theorem}[$\star$]\label{affected-voters-fpt}
\SBON is fixed-parameter tractable parameterized by 
$\lceil \tfrac{n+1}{2}\rceil - \ell$, 
where $\ell$ is the number of voters that already rank the preferred candidate~$c_2$ on top, when there are two candidates, the edge weights are uniform, and the bribery costs are arbitrary.
\end{theorem}
\longversion{
\begin{proof}
Let $n$ denote the number of voters and $\ell$ the number of voters who already rank the preferred candidate $c_2$ on top. 
The goal is to achieve a majority, i.e., to gain at least
\[
p \;=\; \Bigl\lceil \frac{n+1}{2} \Bigr\rceil - \ell
\]
additional supporters.

Consider the underlying network graph $G=(V,E)$, where each voter is a vertex and edges indicate influence. Bribing a voter allows her to support $c_2$ and potentially influence all of her neighbors. This can be naturally modeled as a \((k,t)\)-Dominating Set problem on $G$: we seek a set of at most $k$ vertices (the bribed voters) whose closed neighborhoods collectively dominate at least $t$ vertices, where and $k$ is constrained by the budget (or the maximum number of voters we are allowed to bribe).

It is known that \ktDSP is fixed-parameter tractable when parameterized by $t$~\cite{kneis2007partial}. By applying this result, we can efficiently select a set of voters to bribe so that at least $p$ additional voters are influenced. Therefore, \SBON is fixed-parameter tractable when parameterized by 
\(\lceil (n+1)/2 \rceil - \ell\) in general graphs with two candidates, uniform edge weights, and arbitrary bribery costs.
\end{proof}
}
These parameterized algorithms demonstrate that \SBON can be solved efficiently on networks with favorable structural properties, such as low treewidth or small feedback vertex sets. However, not all parameters lead to tractability. In the next section, we investigate the limits of parameterized approaches by establishing hardness results for \SBON with respect to natural parameters like graphs on bounded degree and budget. This analysis highlights which aspects of the network structure or election setting make the problem inherently difficult, even when the parameter is small.
\subsection{Parameterized Hardness}

We next show that \SBON is \WTH with respect to natural election parameters such as number of candidates $m$, budget $b$, and graphs with bounded degree. 

As a direct consequence of \Cref{sbon-npc}, we can show that \SBON is \WTH parameterized by number of candidates. This immediately gives following observations.

\begin{observation}\label{candidates-hard}
    \SBON parameterized by number of candidates is \WTH.
\end{observation}

\begin{observation}\label{bribed-hard}
\SBON parameterized by number of bribed votes is \WTH.
\end{observation}

Now we explore other parameters as budget and graphs with bounded degree. For that we first show a parameterized reduction from \ktDSP.

Now using \Cref{dsp-hard}, we show that \SBON is \WTH  when parameterized by the budget or degree, even for the case of connected graphs with two candidates, identity cost functions and uniform edge weights.
 
\begin{theorem}[$\star$]\label{w2hbudget}
\SBON is \WTH when parameterized by the budget, even for the case of connected graphs with two candidates, identity cost functions and uniform edge weights.
\end{theorem}
\longversion{
\begin{proof}
We give a parameterized reduction from \ktDSP, which is 
\WTH when parameterized by $k$.

Let $(\GG=(V,E),k,t= \left\lceil \frac{|V| + 1}{2} \right\rceil )$ be an instance of \ktDSP.
We construct an instance of \SBON as follows.
The social network is exactly $\GG$, with uniform edge weights equal to~$1$.
There are two candidates, denoted $c_1$ and $c_2$, where $c_2$ is the preferred candidate. Each voter $v \in \VV$ has the initial preference order $c_1 \succ c_2$, 
so candidate $1$ is last in every vote.
The bribery cost function is uniform: shifting candidate $1$ forward by
one position in any vote costs~$1$.
The budget is set to $b=k$.
The winning condition requires that candidate $1$ is shifted to the top
position in at least $t=\left\lceil \frac{|V| + 1}{2} \right\rceil$ votes after bribery 
and propagation.

Now suppose we bribe a set $D \subseteq V$ of voters, with $|D|\leq k$.
Each such voter immediately shifts candidate $1$ to the top.
Furthermore, since edge weights are uniform and equal to~$1$, every neighbor 
of a bribed vertex also shifts candidate $1$ to the top.
Thus, the set of voters supporting candidate $1$ after propagation is precisely
the dominated set $N[D]\cup D$ in $\GG$.
Therefore, at least $t$ voters support candidate $1$ if and only if 
$D$ dominates at least $t$ vertices in $\GG$.

Hence $(\GG,k,t)$ is a \yes-instance of \ktDSP if and only if the constructed
instance is a \yes-instance of \SBON.
The reduction is polynomial-time and parameter-preserving, as the budget
$b$ equals $k$ 
Therefore, \SBON is \WTH parameterized by the budget, even with two candidates, identity cost functions and uniform edge weights.
\end{proof}
}
Similar to  \Cref{w2hbudget}, using \Cref{dsp-hard} we obtain the following result.

\begin{corollary}\label{w2hdegree}
\SBON is \WTH when parameterized by the maximum degree of the graph even for the case of connected graphs with two candidates, identity cost functions and uniform edge weights.
\end{corollary}

 The reduction in \Cref{sc_equi} is parameter-preserving, as the bribery budget $b$ equals the set cover size $k$. Thus, \Cref{sc_equi} leads us to the following theorem.

\begin{theorem}\label{budget-hard-dibi}
\SBON is \WTH parameterized by the budget $b$, even for undirected bipartite graphs 
with two candidates, identity cost functions and uniform edge weights for the majority voting rule.
\end{theorem}

The hardness hold when restricted to directed bipartite graphs with edges from $\FF$ to $\UU$ in the construction described in \Cref{budget-hard-dibi}. This leads us to the following result.

\begin{corollary}
\SBON is \WTH when parameterized by the budget $b$, even for direct acyclic graphs  with two candidates, uniform edge weights, and arbitrary voting profiles.    
\end{corollary}

\section{Conclusion}\label{sec:conclusion}

We introduced and studied the \SBON problem, which captures shift bribery when voters are embedded in a social influence network. Our results cover both tractable and intractable regimes. On the positive side, we gave an exact dynamic programming algorithm for graphs of bounded treewidth. On the negative side, we showed that \SBON is NP-hard and also \WTH under natural parameterizations like maximum degree of the graph and budget. These hardness results naturally extend to other common voting rules such as Borda and Copeland, where the same structural barriers remain.  
This makes it an interesting direction for future work to design polynomial-time approximation algorithms or fixed-parameter approximation schemes for \SBON and its extensions under richer voting rules.

\section*{Acknowledgments}
We acknowledge the helpful discussions with Abhinav Sen and Shruti Thiagu.

\bibliographystyle{ACM-Reference-Format} 
\bibliography{references}
\clearpage

\shortversion{\section*{Supplementary Material}

This is a Supplementary Material for submission ID: 502. We provide the required technical preliminaries and provide missing proofs of some of the theorems.

\section{Technical Preliminaries}

\textbf{Parameterized Complexity.} In decision problems where the input has size \( n \) and is associated with a parameter \( k \), the objective in parameterized complexity is to design algorithms with running time \( f(k) \cdot n^{O(1)} \), where \( f \) is a computable function that depends only on \( k \) \cite{downey2012parameterized}. Problems that admit such algorithms are said to be \emph{fixed-parameter tractable} (\FPT). An algorithm with running time \( f(k) \cdot n^{O(1)} \) is called an \FPT algorithm, and the corresponding running time is referred to as an \FPT running time \cite{hartmanis2006texts}. Formally, a parameterized problem \( L \subseteq \Sigma^* \times \mathbb{N} \) is said to be \emph{fixed-parameter tractable} if there exists an algorithm \( A \), a computable function \( f : \mathbb{N} \rightarrow \mathbb{N} \), and a constant \( c \in \mathbb{N} \) such that for every input \( (x, k) \in \Sigma^* \times \mathbb{N} \), the algorithm \( A \) decides whether \( (x, k) \in L \) in time at most \( f(k) \cdot |(x, k)|^c \). The class of all such problems is denoted by \FPT \cite{cygan2015parameterized}. 

The treewidth of a graph measures how closely the graph resembles a tree \cite{cygan2015parameterized}. Informally, a tree decomposition of a graph is a tree where each node corresponds to a subset of vertices, called bags, and must satisfy three conditions: (i) every vertex of the graph appears in at least one bag, (ii) both endpoints of every edge are contained in some bag, and (iii) for any vertex, the nodes of the tree containing it must form a connected subtree. For formal definitions of a tree decomposition, a nice tree decomposition, and treewidth, we refer to \cite{cygan2015parameterized}.

\section{Missing Proofs}
\subsection{Proof of \Cref{sbon-npc}}
\begin{proof}
We prove this by a polynomial-time reduction from \DSP, which is known to be \NPC. Given an instance $(\GG, k)$ of \DSP, we construct, in polynomial time, an equivalent instance $(\CC, \PP, \GG_1, c, \Pi, b)$ of \SBON.

Let $\CC = \{c_1, c_2\}$ be the set of candidates, where $c_2$ is the preferred candidate $c$ of the briber. Each vertex $v_i \in \VV$ of $\GG$ corresponds to a voter $v_i$ in the election. We define a directed network $\GG_1 = (\VV_1, E_1, w)$ as follows. We define $\VV_1= \VV \cup \VV_I$, where $\VV_I$ induces an independent set of size $n-1$ in $\GG_1$, and $E_1=E$. Finally, we put weight one on each edge in $E$.


Every voter strictly prefers $c_1$ to $c_2$, i.e., for all $v_i \in \VV_1$, 
$c_1 \succ_i c_2$, at the beginning. Each voter $v_i$ has an identity cost function $\pi_i(s_i) = s_i$, meaning that shifting $c_2$ left by one position (to the top) in $v_i$’s preference order $\succ^{\PP}_{i}$ costs exactly one unit. We set the total bribery budget $b = k$. 

Hence, the constructed instance $(\CC, \PP, \GG_1, c_2, \Pi, b)$ of \SBON can clearly be computed in polynomial time from $(\GG, k)$.
We now argue that the instance $(\GG, k)$ of \DSP is a \yes-instance if and only if the constructed instance $(\CC, \PP, \GG_1, c_2, \Pi, b)$ of \SBON is a \yes-instance.

Suppose $(\GG, k)$ is a \yes-instance of \DSP, with $D \subseteq \VV$ as solution of size at most $k$. We construct a bribery strategy by bribing exactly the voters corresponding to vertices in $D$, shifting $c_2$ to the top of their preference vectors (i.e., $s_i = 1$ for all $v_i \in D$ and $s_i = 0$ otherwise). Since $\pi_i(s_i) = s_i$ and $|D| \le b$, the total bribery cost is within the budget.

Since $D$ is a dominating set in $\GG$, every vertex in $\VV \setminus D$ has at least one bribed neighbor. Moreover, each edge in $\GG_1$ (same as the edges in $\GG$) has weight one, hence, for each unbribed voter $v_j$, the accumulated influence from its bribed neighbors ensures that its effective shift $s_j' \ge 1$, making $c_2$ move to the top of its preference vector. Consequently, every voter in the original graph $\GG$ now ranks $c_2$ above $c_1$. As there are $n$ such voters in $\GG$ and $2n - 1$ voters in total (including the isolated ones), $c_2$ receives at least $n$ votes and becomes the winner. Hence, the constructed instance of \SBON is a \yes-instance.

Conversely, suppose that the constructed \SBON instance is a \yes-instance. Then there exists a shift vector $\sss$ such that the total cost $\sum_i \pi_i(s_i) \le b = k$, and $c_2$ becomes the winner after applying $\sss$. Let $\XX \subseteq \VV_1$ be the set of voters directly bribed (i.e., those with $s_i > 0$). Let $\XX_G = \XX \cap \VV$ denote the bribed voters corresponding to the original graph, and $\XX_I = \XX \setminus \XX_G$ the bribed isolated voters.

Define $\bar{\XX}$ to be the set of voters in $\VV$ who are either in $\XX_G$ or influenced by a voter in $\XX_G$ (i.e., $\bar{\XX} = \XX_G \cup N(\XX_G)$). Let $L = \VV \setminus \bar{\XX}$ denote the set of vertices that are neither bribed nor influenced.

Since $c_2$ wins, at least $n$ voters must now rank $c_2$ above $c_1$, implying $|\bar{\XX}| + |\XX_I| \ge n$. As $\bar{\XX}$ and $L$ partition $\VV$, we have $|\bar{\XX}| + |L| = n$, and thus $|L| \le |\XX_I|$. Because the total bribery cost is at most $k$, we have $|\XX_G| + |\XX_I| \le k$. Combining these inequalities gives
\[
|L| \le |\XX_I| \implies |\XX_G| + |L| \le k.
\]
By construction, $\XX_G \cup L$ forms a dominating set of $\GG$ (since every vertex in $L$ is explicitly chosen to cover the remaining undominated vertices in $\GG$). Therefore, $\GG$ has a dominating set of size at most $k$, and $(\GG, k)$ is a \yes-instance of \DSP.
\end{proof}

\subsection{Proof of \Cref{npc-complete}}
\begin{proof}
Suppose $(\GG, k)$ is a \yes-instance of \DSP. Then there exists a dominating set $D \subseteq \VV$ such that $|D| \le k$. We bribe exactly the voters corresponding to the vertices in $D$, shifting $c_2$ to the top of their preference orders (i.e., $s_i = 1$ for $v_i \in D$ and $s_i = 0$ otherwise). Since $\pi_i(s_i) = s_i$ for these voters, the total bribery cost is at most $|D| \le b$, satisfying the budget constraint.

Since the vertices of $\GG$ form a dominating set, every vertex $v_j \in \VV \setminus D$ has at least one bribed neighbor $v_i \in D$. In $\GG_1$, the influence of each bribed vertex is weighted $w(v_i, v_j) = 1$, so every unbribed voter in $\VV$ experiences a net shift $s_j' \ge 1$ and consequently changes their top preference to $c_2$. Therefore, all $n$ voters corresponding to the original vertices of $\GG$ now prefer $c_2$ over $c_1$. Since there are $2n - 1$ voters in total, $c_2$ secures at least $n$ votes and becomes the winner. Hence, the constructed \SBON instance is a \yes-instance.

Conversely, suppose that the constructed \SBON instance is a \yes-instance. Then there exists a shift vector $\sss$ such that $\sum_i \pi_i(s_i) \le b = k$ and $c_2$ becomes the winner after the shifts are applied. Let $\XX = \{v_i \in \VV_1 \mid s_i > 0\}$ denote the set of bribed voters.

We first note that bribing any voter $v_i \in \VV_I$ (from the added set) costs at least $(k + 1)$, which exceeds the total available budget $k$ i.e. all influence must originate from voters in the original vertex set $\VV$. Moreover, to change such a voter’s effective preference via neighbors would require a total influence of at least $1$, meaning an aggregate shift contribution of at least $2k$ from other voters (since the connecting edges have weight $\frac{1}{2k}$), again exceeding the budget. Hence, no voter in $\VV_I$ can have their preference changed directly or through influence. 

Let $\XX_G = \XX \cap \VV$ denote the bribed voters corresponding to the original graph $\GG$. Because $\sum_i \pi_i(s_i) \le k$, we have $|\XX_G| \le k$. If $\XX_G$ is a dominating set of $\GG$, then $(\GG, k)$ is a \yes-instance of \DSP, as required.

For the sake of contradiction, assume that $\XX_G$ is not a dominating set. Then there exists some vertex $v_j \in \VV$ that is (not dominated), neither bribed nor adjacent to any bribed vertex in $\GG$. Since any pair non-adjacent vertices in $\GG$, connected via an an edge of weight $\frac{1}{2k}$ in $\GG_1$, the influence from all bribed vertices is at most $\frac{|\XX_G|}{2k} \le \frac{1}{2} < 1$, implying $s_j' < 1$. Thus, $v_j$ continues to prefer $c_1$ to $c_2$. Moreover, as established earlier, all $n-1$ voters in $\VV_I$ also continue to prefer $c_1$. Therefore, at most $n - 1$ voters prefer $c_2$, contradicting the assumption that the bribery scheme made $c_2$ win. Hence, $\XX_G$ must be a dominating set of $\GG$ of size at most $k$.
\end{proof}

\subsection{Proof of \Cref{c_wins}}
\begin{proof}
    As discussed previously, all voters in this problem framework shift by the same maximum amount $\alpha$. Let us say c wins for a given value of $\alpha$. This means that c has more votes than all other candidates $c ^\pr \in \CC$. We prove that this cannot change upon increasing $\alpha$. Notice that increasing the value of $\alpha$ cannot reduce the number of votes c gets, since all voters now shift c to a more preferable location. Similarly, notice that increasing $\alpha$ cannot increase the number of votes a rival of c can get, since all other candidates either stay stationary, or shift to a less preferable profile. This completes the proof.
\end{proof}

\subsection{Proof of \Cref{c_looses}}
\begin{proof}
    Very analogous to the previous proof. Let us say c loses for a given value of $\alpha$. This means that c has less votes than some other candidate $c ^\pr \in \CC$. We prove that this cannot change upon decreasing $\alpha$. Notice that decreasing the value of $\alpha$ cannot increase the number of votes c gets, since no voter will shift c to a more preferable location. Similarly, notice that decreasing $\alpha$ cannot decrease the number of votes a rival of c can get, since all other candidates either stay stationary, or shift to a more preferable profile. 
\end{proof}

\subsection{Proof of \Cref{plurality_score}}
\begin{proof}
    For simplicity, let us add 0 to our position vector
    Consider a value of $\alpha$ that is not equal to any $\succ_{i}^{\PP}\left( c \right)$ for any $i \in [n]$ (or 0). Let us assume that this value of $\alpha$ is optimal (the lowest among all $\alpha$ that make c win). Now, consider the largest value of an element in our position vector that is smaller than $\alpha$. Let us call this element p. Notice that reducing the value of $\alpha$ to p cannot cause out candidate c to lose. This is because all voters with $\succ_{i}^{\PP}\left( c \right)$ less than p voted for c previously, and will continue to do so. All the other voters didn't vote for c even with the old value of $\alpha$ since their $\succ_{i}^{\PP}\left( c \right)$ was strictly greater than $\alpha$. This means that $p \geq \succ_{i}^{\PP}\left( c \right) \iff \alpha \geq \succ_{i}^{\PP}\left( c \right)$, which implies that the value p is also a valid $\alpha$ that can make c win. since p $\leq \alpha$, this contradicts the assumption that $\alpha$ is optimal.    \\\\
    Note that this proof does not work if we do not add 0 to our position vector, since $\alpha = 0$ could also be optimal (c is already the winner). 
\end{proof}

\subsection{Proof of \Cref{cluster-graph-pseudo-poly}}
\begin{proof}
Let $\CC = \{c_1,c_2\}$ with $c_2$ being our preferred candidate. Let $\GG=(V,E)$ be a cluster graph consisting of disjoint cliques 
$C_1,\dots,C_r$. 
Each voter $v \in V$ has an arbitrary cost $c(v)$ to shift the preferred candidate~$c_2$ to the top. 
For each clique $C_i$, define
\[
\begin{aligned}
c(C_i) \;&=\; \min_{v \in C_i} c(v), \\[4pt]
gain(C_i) \;&=\; |C_i| \;-\; 
   \#\{\,v \in C_i \mid v \text{ already ranks } c_2 \text{ on top}\,\}.
\end{aligned}
\]

That is, $c(C_i)$ is the minimum cost to flip $C_i$, and $gain(C_i)$ is the number of additional supporters gained if $C_i$ is flipped.

The problem reduces to selecting a subset of cliques whose total cost does not exceed $b$ and whose total gain, together with the initial supporters, reaches the majority threshold $\lceil (n+1)/2 \rceil$.

We construct the following dynamic programming table:
\[
DP[i,d] \;=\; 
\max \Bigl\{ \sum_{j \in S} gain(C_j) \;\Bigm|\; 
S \subseteq \{1,\ldots,i\}, \;
\sum_{j \in S} c(C_j) \leq d \Bigr\}.
\]
That is, $DP[i,d]$ stores the maximum number of additional supporters obtainable from the first $i$ cliques using budget at most $d$.

The recurrence is
\(
DP[i,d] \;=\; 
\max \bigl(
DP[i-1,d], \;
DP[i-1,d - c(C_i)] + gain(C_i)
\bigr),
\)
whenever $d \geq c(C_i)$, and otherwise
\(
DP[i,d] \;=\; DP[i-1,b].
\)

The base case is $DP[0,d]=0$ for all $d \geq 0$.  
At the end, let $l$ be the number of initial supporters of candidate~$c_2$.  
We check whether
\(
\max_{d\leq b} DP[r,d] + l \;\;\geq\;\; \lceil (n+1)/2 \rceil.
\)
If yes, then the bribery succeeds within budget $b$; otherwise, it is impossible.

The DP table has size $O(r \cdot b)$ and each entry can be computed in constant time from previous entries. Thus, the algorithm runs in $O(r\cdot b)$ time, which is pseudo-polynomial in the budget $b$. 
\end{proof}

\subsection{Proof of \Cref{sc_equi}}
\begin{proof}
Suppose first that the given \textsc{Set Cover} instance is a \yes-instance, and let $\{S_{i_1},\ldots,S_{i_k}\}$ be a cover of size $k$. Bribe the corresponding vertices $\{r_{i_1},\ldots,r_{i_k}\}\subseteq R_{\mathrm{orig}}$. Then all $n$ vertices in $L_{\mathrm{orig}}$ are influenced, since the chosen subfamily covers $\mathcal{U}$. Moreover, all $m-k+1$ vertices in $L_{\mathrm{new}}$ are influenced, as they are adjacent to every vertex in $R_{\mathrm{orig}}$. Finally, the $k$ bribed vertices in $R_{\mathrm{orig}}$ themselves support candidate~$c_2$. Therefore, the total number of supporters of $c_2$ is $n+(m-k+1)+k = n+m+1$, meeting the majority threshold. Hence the constructed \SBON instance is a \yes-instance.

Conversely, suppose that the constructed \SBON instance is a \yes-instance. Then there exists a bribery set $B$ with $|B|\leq k$ such that at least $n+m+1$ voters support candidate~$c_2$ after bribery and propagation. Partition $B$ as
\[
B=(L_1\cup L_2)\cup(R_1\cup R_2),
\]
where $L_1\subseteq L_{\mathrm{orig}}$, $L_2\subseteq L_{\mathrm{new}}$, 
$R_1\subseteq R_{\mathrm{orig}}$, and $R_2\subseteq R_{\mathrm{new}}$.

\begin{claim}
There exists a solution with $B\subseteq R_{\mathrm{orig}}$.
\end{claim}

\begin{proof}
We first argue that an optimal bribery strategy of cost at most $k$ may, without loss of generality, be assumed to bribe only vertices in $R_{\mathrm{orig}}$. Indeed, bribing a vertex of $L_1$ influences only that vertex, whereas bribing one of its neighbors in $R_{\mathrm{orig}}$ influences not only that vertex of $L_1$ but possibly additional vertices as well. Hence any bribery that targets a vertex of $L_1$ can be replaced by bribing one of its neighbors in $R_{\mathrm{orig}}$ without decreasing the number of supporters. Similarly, bribing a vertex of $L_2$ is redundant, since every vertex of $L_2$ is already influenced whenever any neighbor in $R_{\mathrm{orig}}$ is bribed. Finally, bribing a vertex of $R_2$ yields only that single supporter, since these vertices are isolated, and replacing such a bribery with one targeting a vertex of $R_{\mathrm{orig}}$ never decreases the total number of influenced vertices. 
\end{proof}

Next, consider the requirement of reaching $n+m+1$ supporters. Since all $m-k+1$ vertices of $L_{\mathrm{new}}$ must be influenced, and each such vertex has neighbors only in $R_{\mathrm{orig}}$, this necessitates that $B$ contain at least one vertex of $R_{\mathrm{orig}}$. Moreover, all $n$ vertices of $L_{\mathrm{orig}}$ must be influenced as well, for otherwise the total number of supporters is strictly less than $n+m+1$. Consequently, every vertex of $L_{\mathrm{orig}}$ is adjacent to some bribed vertex in $R_{\mathrm{orig}}$, which means that $B$ dominates all of $L_{\mathrm{orig}}$ and satisfies $|B|\leq k$. By construction of the reduction, each vertex of $R_{\mathrm{orig}}$ corresponds to a set in the family of the original set cover instance, and domination of $L_{\mathrm{orig}}$ by $B$ corresponds precisely to covering the universe $\mathcal{U}$. Therefore, the bribery strategy $B$ induces a set cover of $\mathcal{U}$ of size at most $k$.
\end{proof}

\subsection{Proof of \Cref{treewidth-correctness}}
\begin{proof} We maintain the following invariant: for every node $t$ of the nice tree decomposition with bag $X_t$, for every $S\subseteq X_t$, $I\subseteq X_t$, and $0\le b'\le b$, the entry
\(
DP[t,S,I,b']
\)
equals the maximum number of vertices in $V_t$ that can be influenced by any bribery configuration restricted to the subtree rooted at $t$ that
\begin{enumerate}
  \item bribed exactly the vertices $S$ inside the bag $X_t$,
  \item results in exactly the set $I$ of influenced vertices inside the bag $X_t$, and
  \item spends total bribery cost exactly $b'$ inside $V_t$.
\end{enumerate}
If no such configuration exists, $DP[t,S,I,b']=-\infty$.

We prove the invariant by induction on the height (distance from the root) of node $t$.

\paragraph{\textbf{Base (leaf)}.}
If $t$ is a leaf then $X_t=\varnothing$ and $V_t=\varnothing$. The only feasible configuration uses zero budget and influences zero vertices; hence
\(
DP[t,\varnothing,\varnothing,0]=0,
\)
and all other entries are infeasible ($-\infty$). This matches the invariant.

\textbf{Inductive step.}
Assume the invariant holds for all children of $t$. We show each node-type transition computes the correct maximum and therefore the invariant holds at $t$.
\begin{enumerate}
    \item Introduce node. Suppose $t$ introduces vertex $v$ and its child is $t'$ with $X_t = X_{t'}\cup\{v\}$. Fix $(S,I,b')$ for $t$. Any feasible completion of this partial assignment in $V_t$ induces a feasible completion in $V_{t'}$ with bag assignment $(S',I')$ where $S' = S\setminus\{v\}$ and $I' = I\cap X_{t'}$. There are exactly three mutually exclusive possibilities for $v$ relative to the bribery/influence decisions:
\begin{enumerate}
  \item $v\in S$ (we bribed $v$). Then $v$ must be in $I$, we pay cost $c(v)$, and bribing $v$ may additionally influence neighbours of $v$ that lie in $X_{t'}$. Every feasible completion in $V_t$ corresponds to a feasible completion in $V_{t'}$ with $S',I'$ where neighbours already counted are removed as in the recurrence. By the inductive hypothesis the child DP value yields the maximum number of influenced vertices in $V_{t'}$ consistent with that child assignment; adding $v$ and the newly influenced neighbours yields exactly the total in $V_t$. The recurrence
   \[
    DP[t,S,I,b'] \;=\; 
    DP[t',S',I',b'-c(v)] + 1 
    + \big|\{ u \in N(v)\cap X_{t'} : u \notin I'\}\big|.
    \]
  therefore computes the maximum over all completions in this case (and is $-\infty$ when $b'-c(v)<0$).
  \item $v\notin S$ but $v\in I$ (not bribed but already influenced). Then some bribed neighbour in $V_t$ must influence $v$. Any such completion restricts to a child completion with $(S',I',b')$ and the validity of $v\in I$ is exactly captured by the indicator $\mathbf{1}_{\{\exists u\in N(v)\cap X_{t'}:u\in S'\}}$ (or more generally, check for bribed neighbours in the processed part). By the inductive hypothesis, $DP[t',S',I',b']$ is the maximum achievable in the child; adding the one for $v$ when allowed gives the correct value.
  \item $v\notin S$ and $v\notin I$. Then $v$ contributes nothing and feasible completions correspond one-to-one with child completions having $(S',I',b')$, so
  \[
  DP[t,S,I,b'] = DP[t',S',I',b']
  \]
  is correct.
\end{enumerate}
Since these three cases are exhaustive and mutually exclusive, taking the maximum over them yields the correct optimal value for $DP[t,S,I,b']$.

\item Forget node. Suppose $t$ forgets vertex $v$ so $X_t = X_{t'}\setminus\{v\}$. Fix $(S,I,b')$ at $t$. Any feasible completion in $V_t$ arises from some feasible completion in $V_{t'}$ where $v$ had one of the three local statuses in the child bag: (i) bribed and influenced, (ii) not bribed but influenced, or (iii) not bribed and not influenced. These correspond precisely to the three child states
\[
(S\cup\{v\},I\cup\{v\},b'),\quad (S,I\cup\{v\},b'),\quad (S,I,b'),
\]
respectively. By the inductive hypothesis each child DP value is the maximum over completions in the child consistent with that child state, and since nothing additional is gained at the forget operation (vertices are already counted when influenced), the maximum of these three child values equals the maximum over all completions in $V_t$ consistent with $(S,I,b')$. Hence the recurrence
\[
DP[t,S,I,b'] \;=\;
\max\left\{
\begin{aligned}
   &DP[t',S\cup\{v\},I\cup\{v\},b'], \\
   &DP[t',S,I\cup\{v\},b'], \\
   &DP[t',S,I,b']
\end{aligned}
\right\}.
\]
is correct.

\item Join node. Suppose $t$ has two children $t_1,t_2$ with $X_{t_1}=X_{t_2}=X_t$. Fix $(S,I,b')$ at $t$. Any feasible completion in $V_t$ decomposes into completions in the two child subtrees whose bag-assignments agree with $(S,I)$ and whose budgets sum to $b'$. Each child DP (by the inductive hypothesis) equals the maximum number of influenced vertices in its subtree consistent with the child assignment. However, vertices in the bag $X_t$ (in particular the influenced set $I$) are counted in both child DP values, so we must subtract $|I|$ once to avoid double counting. Maximizing over all budget splits $b_1+b_2=b'$ yields:
\[
DP[t,S,I,b'] = \max_{b_1+b_2=b'} \big( DP[t_1,S,I,b_1] + DP[t_2,S,I,b_2] - |I| \big),
\]
which therefore gives the correct maximum for $V_t$.

\end{enumerate}
\medskip\noindent\textbf{Root.} Let $r$ be the root with bag $X_r$. The whole graph equals $V_r$. By the invariant, for each $(S,I,b')$ the value $DP[r,S,I,b']$ equals the maximum number of influenced vertices achievable in the whole graph under those bag constraints; hence taking
\[
\max_{S,I,\,b'\le b} DP[r,S,I,b']
\]
gives the optimal number of influenced vertices for the input instance. This completes the induction and correctness proof. Now it remains to check whether \(
\max_{S,I,\,b'\le b} DP[r,S,I,b'] \geq \frac{n}{2} + 1
\)
or not.

\paragraph{Reconstruction.}
Since we only have two candidates, every bribery corresponds to 
shifting our preferred candidate one position upwards in the vote. 
Thus, the shift vector is binary: for each voter $v$, 
\[
s_v \;=\;
\begin{cases}
1 & \text{if $v$ was bribed}, \\
0 & \text{otherwise}.
\end{cases}
\]
Using standard Dynamic programming techniques, we can track which vertices are bribed.
\end{proof}

\subsection{Proof of \Cref{twproof}}
\begin{proof}
Let $\tw := |X_t| \le \operatorname{tw}(\GG)+1$ denote the (maximum) bag size. For each node $t$ we store an entry for every triple $(S,I,b')$ where $S\subseteq X_t$, $I\subseteq X_t$, and $b'\in\{0,1,\dots,b\}$. Hence the number of states per node is at most
\(2^{\tw}\cdot 2^{\tw}\cdot (b+1) = 4^{\tw}(b+1).
\)
Since the treewidth of the input graph $\GG$ is at most $\tw$, 
it is possible to construct a data structure in time $\tw^{\OO(1)} \cdot n$ 
that allows performing adjacency queries in time $\OO(\tw)$. 
As we may assume without loss of generality that the number of nodes in the 
tree decomposition is $\OO(\tw \cdot n)$, and there are $n$ choices for the 
introduced vertex $v$, the running time of the algorithm is
\(
\OO\!\left(4^{\tw}\cdot n^{\OO(1)} \cdot b\right).
\)

Since we only have two candidates, every bribery corresponds to a unit shift, 
and hence each bribed voter contributes cost~1. Therefore $b$ can be bounded 
as follows:
\(
b \;\leq\; n-\ell,
\)
where $\ell$ is the number of voters already voting for the preferred 
candidate. Moreover, if the goal is to ensure victory of the preferred 
candidate, it suffices to consider
\[
b \;\leq\; \kappa \;=\; 
\max\Bigl\{0,\; \bigl\lceil\tfrac{n}{2}\bigr\rceil+1 - \ell\Bigr\},
\]
since $\kappa$ is exactly the number of additional votes required to reach a 
strict majority. 

Thus, the overall running time of our algorithm is
\(
\OO\!\bigl(4^{\tw}\cdot n^{\OO(1)} \cdot \kappa\bigr),
\)
where $\kappa \leq \lceil (n + 1)/2 \rceil$.
\end{proof}

\subsection{Proof of \Cref{cvd-fpt}}
\begin{proof}
Let $\GG=(V,E)$ be the underlying graph of the given instance of \SBON. 
Suppose that $X \subseteq V$ is a cluster vertex deletion set of size $k$, i.e.\ $\GG - X$ is a cluster graph. 
It is known that such a set $X$ can be computed in FPT time with respect to~$k$`\cite{boral2016fast}.

Since $k$ is constant, we can exhaustively try all subsets $Y \subseteq X$ that may be part of the bribery scheme. 
For each choice of $Y$, we compute the cost of bribing $Y$ and the gain obtained in terms of additional supporters of the preferred candidate. 
If the cost of bribing $Y$ already exceeds the budget, we discard this choice.

Now consider the residual instance on $\GG - X$, which is a cluster graph. 
The effect of bribing vertices in $Y$ is that some vertices of $\GG - X$ may already support the preferred candidate (either directly if bribed, or indirectly via propagation). 
Thus, the residual problem on $\GG - X$ is precisely an instance of \SBON on a cluster graph with updated initial supporters and reduced budget. 
This can be solved in pseudo-polynomial time $O(r \cdot b)$ using the dynamic programming algorithm from the previous theorem, where $r$ is the number of cliques in $\GG - X$ and $b$ is the remaining budget.

Since the number of subsets $Y \subseteq X$ is $2^k$ and $k$ is a constant, this branching contributes only a constant factor. 
Hence the overall running time is $O(2^k \cdot poly(r \cdot b)$. 
For constant $k$, this is polynomial in the input size. 
More generally, this shows that the problem is pseudo fixed-parameter tractable when parameterized by $k$.
\end{proof}

\subsection{Proof of \Cref{affected-voters-fpt}}
\begin{proof}
Let $n$ denote the number of voters and $\ell$ the number of voters who already rank the preferred candidate $c_2$ on top. 
The goal is to achieve a majority, i.e., to gain at least
\[
p \;=\; \Bigl\lceil \frac{n+1}{2} \Bigr\rceil - \ell
\]
additional supporters.

Consider the underlying network graph $G=(V,E)$, where each voter is a vertex and edges indicate influence. Bribing a voter allows her to support $c_2$ and potentially influence all of her neighbors. This can be naturally modeled as a \((k,t)\)-Dominating Set problem on $G$: we seek a set of at most $k$ vertices (the bribed voters) whose closed neighborhoods collectively dominate at least $t$ vertices, where and $k$ is constrained by the budget (or the maximum number of voters we are allowed to bribe).

It is known that \ktDSP is fixed-parameter tractable when parameterized by $t$~\cite{kneis2007partial}. By applying this result, we can efficiently select a set of voters to bribe so that at least $p$ additional voters are influenced. Therefore, \SBON is fixed-parameter tractable when parameterized by 
\(\lceil (n+1)/2 \rceil - \ell\) in general graphs with two candidates, uniform edge weights, and arbitrary bribery costs.
\end{proof}

\subsection{Proof of \Cref{w2hbudget}}
\begin{proof}
We give a parameterized reduction from \ktDSP, which is 
\WTH when parameterized by $k$.

Let $(\GG=(V,E),k,t= \left\lceil \frac{|V| + 1}{2} \right\rceil )$ be an instance of \ktDSP.
We construct an instance of \SBON as follows.
The social network is exactly $\GG$, with uniform edge weights equal to~$1$.
There are two candidates, denoted $c_1$ and $c_2$, where $c_2$ is the preferred candidate. Each voter $v \in \VV$ has the initial preference order $c_1 \succ c_2$, 
so candidate $1$ is last in every vote.
The bribery cost function is uniform: shifting candidate $1$ forward by
one position in any vote costs~$1$.
The budget is set to $b=k$.
The winning condition requires that candidate $1$ is shifted to the top
position in at least $t=\left\lceil \frac{|V| + 1}{2} \right\rceil$ votes after bribery 
and propagation.

Now suppose we bribe a set $D \subseteq V$ of voters, with $|D|\leq k$.
Each such voter immediately shifts candidate $1$ to the top.
Furthermore, since edge weights are uniform and equal to~$1$, every neighbor 
of a bribed vertex also shifts candidate $1$ to the top.
Thus, the set of voters supporting candidate $1$ after propagation is precisely
the dominated set $N[D]\cup D$ in $\GG$.
Therefore, at least $t$ voters support candidate $1$ if and only if 
$D$ dominates at least $t$ vertices in $\GG$.

Hence $(\GG,k,t)$ is a \yes-instance of \ktDSP if and only if the constructed
instance is a \yes-instance of \SBON.
The reduction is polynomial-time and parameter-preserving, as the budget
$b$ equals $k$ 
Therefore, \SBON is \WTH parameterized by the budget, even with two candidates, identity cost functions and uniform edge weights.
\end{proof}}


\end{document}